\documentclass[runningheads]{llncs}
\usepackage[T1]{fontenc}
\usepackage{graphicx}
\usepackage{cite}
\usepackage{amsmath,amssymb,amsfonts}
\usepackage{booktabs}
\usepackage{tabularx}
\usepackage[colorlinks=true,
            linkcolor={red!70!black},
            citecolor={green!70!black},
            urlcolor={blue!80!black}]{hyperref}
\usepackage[all]{hypcap}
\usepackage{tikz}
\usetikzlibrary{arrows.meta,positioning,fit,shapes.misc,calc,backgrounds,decorations.pathreplacing}
\spnewtheorem{assumption}{Assumption}{\bfseries}{\itshape}

\usepackage{pgfplots}
\pgfplotsset{compat=1.18}
\usepackage{xcolor}
\usepackage{orcidlink}

\definecolor{sciBlue}{HTML}{0077BB}
\definecolor{sciOrange}{HTML}{EE7733}
\definecolor{sciTeal}{HTML}{009988}
\definecolor{sciRed}{HTML}{CC3311}
\definecolor{sciCyan}{HTML}{33BBEE}
\definecolor{sciMagenta}{HTML}{EE3377}
\definecolor{sciPurple}{HTML}{AA4499}
\definecolor{sciGray}{HTML}{555555}

\pgfplotsset{
    every axis/.append style={
        tick label style={font=\small},
        label style={font=\small},
        legend style={font=\small, fill=white, fill opacity=0.9, draw opacity=1, text opacity=1},
        grid style={densely dotted, color=gray!90},
        thick,
        enlargelimits=0.05,
        xtick={1, 2, 3, 4, 5, 6, 7, 8, 9, 10},
    }
}

\newcommand{\Rq}{\mathcal{R}_q}
\newcommand{\NTT}{\mathsf{NTT}}
\newcommand{\sn}{\mathsf{sn}}
\newcommand{\snstar}{\mathsf{sn}^{\ast}}
\newcommand{\HtoC}{\mathsf{H2C}}
\newcommand{\pack}{\mathsf{pack}}
\newcommand{\CoinKeyGen}{\textsf{CoinKeyGen}}
\newcommand{\SerNum}{\textsf{SN}}
\newcommand{\Sig}{\textsf{Sign}}
\newcommand{\Ver}{\textsf{Verify}}
\newcommand{\Link}{\textsf{Link}}

\title{Obscura-PQ: Post-Quantum Privacy-Preserving Protocol for the Algorand Blockchain Using Lattice-Based Linkable Ring Signatures}
\titlerunning{Obscura-PQ: Post-Quantum Privacy for Algorand}
\author{Navid Azimi \orcidlink{0009-0009-1980-8521}}
\authorrunning{N. Azimi}
\institute{Emory University, Atlanta, GA 30322, USA  \\
\email{navid.azimi@emory.edu}
}
\begin{document}
\maketitle
\begin{abstract}
Public blockchains expose the complete transaction graph, and the privacy protocols deployed to obscure it rely almost exclusively on elliptic-curve cryptography, whose discrete-logarithm foundations fall to Shor's algorithm. Because ledgers are immutable, every anonymity set published today under classical assumptions can be retroactively deanonymized by a future quantum adversary. Transitioning to post-quantum alternatives remains challenging, as the strict resource limits of smart contracts prohibit the native on-chain verification of computationally intensive post-quantum proofs. To address these challenges, we present \emph{Obscura-PQ}, a decentralized, non-custodial post-quantum privacy protocol that verifies natively on the Algorand blockchain. Its core is a setup-free lattice linkable ring signature over the cyclotomic ring $\mathcal{R}_q = \mathbb{Z}_q[X]/(X^{512}+1)$. A deposit is a Ring-SIS binding commitment to a short secret; a withdrawal proves knowledge of a ring opening via an AOS/Borromean-style challenge chain over two response-sharing linear relations with rejection-sampled short responses, while publishing a deterministic Ring-LWE serial number for double-spend detection. We reduce double-spend soundness and linkability to the Ring-SIS assumption, theft resistance to Ring-SIS for rings of honestly generated deposits, and anonymity to Ring-LWE and an explicit decisional linking assumption, all in the classical random-oracle model. To overcome strict on-chain opcode and storage limits, Obscura-PQ evaluates all verification relations entirely in the number-theoretic-transform (NTT) domain. We split the forward NTTs across opcode-pooled execution phases and stream oversized proofs through refundable box storage, enabling $O(1)$ membership and double-spend checks. We provide a complete Algorand testnet implementation, demonstrating native on-chain verification of a post-quantum privacy protocol under strict smart-contract limits.
\end{abstract}
%
\section{Introduction}
\label{sec:intro}
The transparency of public blockchains \cite{nakamoto2008bitcoin} underpins their security and auditability, yet it inherently undermines user privacy: adversaries can trace the flow of funds, cluster addresses, and deanonymize users through transaction-graph analysis \cite{ron2013quantitative,meiklejohn2013fistful,androulaki2013evaluating}. Privacy-preserving protocols and decentralized mixers address this issue \cite{miers2013zerocoin,sasson2014zerocash,bonneau2014mixcoin}: on Ethereum \cite{wood2014ethereum}, Tornado Cash \cite{pertsev2019tornado} severs the on-chain link between depositors and withdrawers using zk-SNARKs over a Merkle accumulator \cite{mendonca2026efficient,merkle1987digital}; M\"obius \cite{meiklejohn2018mobius} achieves trustless tumbling from linkable ring signatures; Zether \cite{bunz2020zether} provides confidential account balances; Obscura \cite{azimi2026obscura} deploys a non-custodial LSAG-based privacy protocol natively on Algorand; and Monero's CryptoNote lineage \cite{van2013cryptonote,noether2016ring,goodell2019concise} deploys linkable ring signatures at the base layer.

All of these systems share a single point of failure: their unforgeability or soundness, and the long-term secrecy of their published commitments, ultimately rest on discrete-logarithm or pairing assumptions in elliptic-curve groups, all of which Shor's algorithm breaks \cite{shor1994algorithms}. The threat is not merely the hypothetical future breakage of future transactions. Because a blockchain is an immutable public record, every commitment, ring, and nullifier published under classical assumptions is archived permanently; a quantum adversary arriving decades from now can retroactively recover the secret behind each deposit, link deposits to withdrawals, and dissolve every anonymity set ever formed. Unlike signatures, which can be rotated, privacy protocols face \emph{harvest-now, deanonymize-later} as the default attack, and this can be prevented by making the published transcript post-quantum secure from the outset.

A second, orthogonal obstacle is the execution environment. High-throughput blockchains such as Algorand \cite{gilad2017algorand} preserve fast finality and predictable fees by imposing strict per-transaction opcode budgets, flat state models, and small argument-size limits on their virtual machines \cite{algorand_avm,algorand_sc_costs_constraints}. Lattice-based zero-knowledge machinery is far heavier than its elliptic-curve counterpart: a single polynomial multiplication in a post-quantum ring already exceeds the entire pooled execution budget of a transaction group, and a single lattice commitment is more than an order of magnitude larger than an elliptic-curve point. No deployed smart-contract privacy protocol that we are aware of verifies a post-quantum anonymous spend natively on-chain.

\paragraph{Proposed Approach.}
We present \emph{Obscura-PQ}, the post-quantum version of the Obscura protocol \cite{azimi2026obscura}: a decentralized, non-custodial privacy protocol for Algorand whose entire cryptographic core is post-quantum. Obscura-PQ is a fixed-denomination privacy pool: a user deposits by publishing a \emph{coin}, and later withdraws to a fresh address by proving, without revealing which one, that they own one of the coins in a self-selected \emph{ring} of on-chain deposits, while publishing a \emph{key image} that makes double-spends detectable.

The construction works over the power-of-two cyclotomic ring $\Rq = \mathbb{Z}_q[X]/(X^{n}+1)$ with $n = 512$ and $q = 12289$ (the ring underlying Falcon-512 \cite{fouque2018falcon}) and follows the commitment-and-serial-number paradigm of Lattice RingCT and MatRiCT \cite{alberto2018post,esgin2019matrict}. A coin is a Ring-SIS binding commitment $C = a_1 k + a_2 s + a_4 e$ (a rank-one module commitment in the style of \cite{baum2018more}) to a unique short spend secret $k$, short randomness $s$, and a short serial-number blinding term $e$; the key image is the deterministic \emph{Ring-LWE serial number} $\sn = a_3 k + e$, revealed in full at withdrawal; and the one-out-of-$r$ spend proof is an AOS/Borromean-style challenge chain \cite{abe20041,Maxwell2015BorromeanRS} over \emph{two response-sharing linear relations}, one opening the commitment, one producing the serial number, closed at the true index with rejection-sampled short responses \cite{lyubashevsky2009fiat,lyubashevsky2012lattice}. The two-relation response sharing that binds the serial number to the coin serves as the lattice analogue to the parallel-equation structure of CLSAG \cite{goodell2019concise}. While the compressed response-sharing closure of DualRing \cite{yuen2021dualring} shares a similar intuition, it remains structurally distinct: DualRing employs a ring of challenges with a single response, whereas our construction relies on a sequential challenge chain. The public parameters $a_1, a_2, a_3, a_4$ are nothing-up-my-sleeve hash outputs, so the protocol has no trusted setup. Soundness and linkability reduce to Ring-SIS \cite{lyubashevsky2006generalized,alberto2018post}, and anonymity to Ring-LWE \cite{lyubashevsky2010ideal}, in the random-oracle model.

\paragraph{Key Challenges and Technical Solutions.}
Two problems separate this design from a routine instantiation.

First, a \emph{sound post-quantum key image}. A key image supports double-spend detection only if it satisfies three properties at once: it must be \emph{deterministic}, so that every spend of a given coin yields the same image; it must be derived from a \emph{unique} secret, so that a malicious signer cannot equip one coin with several valid images; and it must be \emph{member-hiding}, so that publishing it reveals nothing about which ring member it corresponds to. Obscura-PQ derives all three properties from commitment binding. The uniqueness of the spend secret $k$ and the blinding term $e$ is guaranteed by the Ring-SIS assumption, as computing a second short opening for the commitment $C$ would yield a Ring-SIS collision. This unique opening deterministically fixes the serial number $\sn = a_3 k + e$. Revealing $\sn$ in the clear compromises neither the secret nor the anonymity of the transaction: because $\sn$ forms a valid Ring-LWE sample (where $k$ acts as the secret and $e = H_{\mathrm{short}}(k)$ as the error), it is pseudorandom in the random-oracle model and computationally unlinkable to $C$. Furthermore, while $e$ is derived deterministically to allow honest signers to remain stateless, the \emph{soundness} of the protocol relies strictly on $e$ being algebraically bound to $C$ via the term $a_4 e$, rather than on the hash derivation. Consequently, a malicious signer who commits using an arbitrary short $e$ is inextricably pinned to that same $e$ for all future spending attempts. Finally, because the signing protocol shares response coordinates across both verification relations, every verifying signature guarantees that one short opening simultaneously opens some ring commitment and yields the published serial number.

Second, \emph{native on-chain lattice verification}. Verifying the signature requires, for each ring member, four forward number-theoretic transforms (NTTs) of degree-512 polynomials, a pointwise relation evaluation, norm checks on the responses, and a hash-chain update. This workload collides with the AVM's resource model on two fronts: a single forward NTT with its norm check costs more than the entire pooled budget of an Algorand transaction group, and the proof, at $32 + 3072r$ bytes for ring size $r$, exceeds the 2\,048-byte application-argument limit for any $r$. Obscura-PQ overcomes both obstacles with four mutually supporting techniques. (i) All relations are evaluated \emph{in the NTT domain}: commitments are stored, and serial numbers transported, as forward transforms, so the verifier never computes an inverse NTT. (ii) Each forward NTT is \emph{split} at a fixed butterfly length into two execution phases, nine phases per ring member in total, each provisioned with pooled opcode budget through inner ``no-op'' application calls. (iii) Proofs are \emph{streamed} into an application-owned transport box before verification begins; all transient storage deposits are fronted by the withdrawer and refunded atomically at settlement. (iv) Commitments and nullifiers reside in a flat key--value \emph{box-storage} model \cite{algorand_box_storage}, which supports $O(1)$ membership and double-spend checks without a global Merkle accumulator.

\paragraph{Contributions.}
This paper makes the following contributions:
\begin{itemize}
    \item \textbf{Protocol design.} We propose Obscura-PQ, a post-quantum privacy protocol whose anonymous spend is a lattice linkable ring signature built from a Ring-SIS binding commitment, a deterministic linear serial number, and an AOS/Borromean-style challenge chain over two response-sharing relations. Every Fiat--Shamir challenge is computed over the full settlement context (recipient, relayer, fee, and application identifier), defeating replay and fee/relayer front-running, and the public parameters are nothing-up-my-sleeve.
    \item \textbf{On-chain execution model for lattice verification.} We show how to run the verifier natively on a strictly budgeted virtual machine: a frequency-domain wire format that eliminates inverse transforms, table-driven negacyclic NTTs split across opcode-pooled phases, box-streamed proof transport with refundable storage accounting, and a per-member phase machine whose state survives across transaction groups. The resulting verifier is correct for every ring size up to the protocol cap.
    \item \textbf{Security analysis.} We give a precise threat model and reductions: theft resistance (unforgeability) and double-spend soundness to Ring-SIS, signer ambiguity and unlinkability to Ring-LWE and an explicit decisional linking assumption for the serial number closed with the rejection-sampling zero-knowledge argument, and replay resistance to context binding, all in the random-oracle model. Because the fully splitting Falcon-512 ring does not admit invertible short challenge differences, our extractor recovers a \emph{relaxed} opening and our unforgeability statement is proved for rings of honestly generated deposits (see Appendix~\ref{app:formal}, Remark~\ref{rem:gaps}).
    \item \textbf{End-to-end implementation.} We provide a complete, open-source full-stack deployment on the Algorand testnet, featuring a PyTeal smart contract compiled to AVM bytecode, a local Python-based prover, and a React client interface.
\end{itemize}

\medskip
\noindent
Figure~\ref{fig:arch} illustrates the end-to-end architecture. The remainder of the paper is organized as follows: Section~\ref{sec:related} surveys related work, and Section~\ref{sec:prelim} establishes the necessary cryptographic and system preliminaries. We formalize the core construction in Section~\ref{sec:crypto} and detail the protocol flows in Section~\ref{sec:protocol}. The system design is then presented in two parts: Section~\ref{sec:onchain} describes the on-chain execution model, while Section~\ref{sec:impl} covers the implementation and its validation. Finally, Section~\ref{sec:sec} analyzes security and privacy, Section~\ref{sec:eval} develops the cost model, and Section~\ref{sec:concl} concludes. Formal definitions and security arguments are deferred to Appendix~\ref{app:formal}.

\begin{figure}[t]
  \centering
  \includegraphics[width=\textwidth]{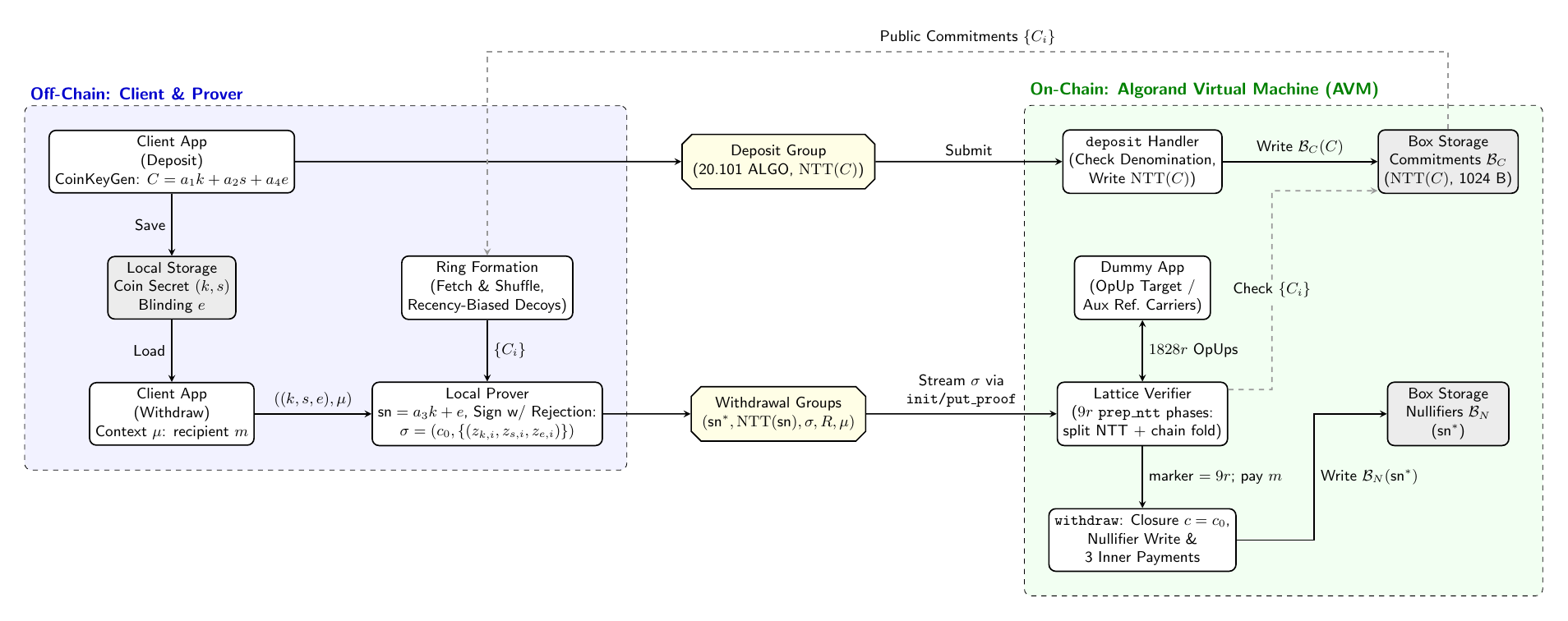}
    \caption{Obscura-PQ protocol end-to-end architecture. Solid arrows denote protocol actions and transaction submissions; dashed arrows denote reads of on-chain public state. The Client App initiates both phases. \textbf{Deposit:} the client runs CoinKeyGen, sampling a short coin secret $(k, s)$ with deterministic blinding $e = H_{\mathrm{short}}(k)$, computes the Ring-SIS commitment $C = a_1 k + a_2 s + a_4 e$, stores $(k, s, e)$ locally, and submits a deposit group $(20.101~\mathrm{ALGO}, \mathrm{NTT}(C))$. The on-chain \texttt{deposit} handler checks the fixed denomination and writes $\mathrm{NTT}(C)$ into commitment Box Storage $\mathcal{B}_C$. \textbf{Withdrawal:} the client loads $(k, s, e)$, specifies context $\mu$ (including recipient $m$), and fetches public commitments $\{C_i\}$ from $\mathcal{B}_C$ to form a recency-biased, shuffled ring $R$. The local prover then computes the serial number $\mathsf{sn} = a_3 k + e$ and, with rejection sampling, a lattice AOS/Borromean-style linkable ring signature $\sigma = (c_0, \{(z_{k,i}, z_{s,i}, z_{e,i})\})$; the digest $\mathsf{sn}^\ast$ is the nullifier. Because $\sigma$ exceeds application-argument limits, the client streams it on-chain via \texttt{init}/\texttt{put\_proof}. The lattice verifier reads $\{C_i\}$ from $\mathcal{B}_C$, expands its opcode budget via $1828r$ inner \texttt{opup} calls to a dummy application, and verifies $\sigma$ natively through $9r$ sequential \texttt{prep\_ntt} phases (split NTT and Fiat--Shamir chain fold). After asserting chain closure $c = c_0$ (marker $= 9r$), the \texttt{withdraw} handler records $\mathsf{sn}^\ast$ in nullifier Box Storage $\mathcal{B}_N$ if unspent and settles with three inner payments to $m$.}
\label{fig:arch}
\end{figure}

\section{Related Work}
\label{sec:related}
\paragraph{Smart-Contract Privacy Pools.}
Tornado Cash \cite{pertsev2019tornado} pioneered SNARK-based privacy pools on Ethereum, achieving pool-wide anonymity via Groth16 proofs \cite{groth2016size} over a Merkle accumulator at the cost of a trusted setup and pairing-heavy verification; Zerocoin and Zerocash \cite{miers2013zerocoin,sasson2014zerocash} established the commitment/nullifier paradigm at the base layer. M\"obius \cite{meiklejohn2018mobius} replaces SNARKs with linkable ring signatures for trustless tumbling, and Zether \cite{bunz2020zether} provides confidential balances; both are elliptic-curve-based. Obscura \cite{azimi2026obscura} brought the ring-signature privacy pool paradigm to Algorand, introducing the flat box-storage state model with $O(1)$ membership checks and the inner-transaction opcode-pooling on which the present work builds; its LSAG core, like all of the above, rests on discrete-logarithm assumptions. Obscura-PQ is its post-quantum evolution: it keeps the setup-free, accumulator-free ring paradigm and the box-based state architecture, but makes the entire published transcript post-quantum, which, as Sections~\ref{sec:onchain} and~\ref{sec:eval} show, changes the execution problem qualitatively, since lattice verification no longer fits within any single transaction group.

\paragraph{Linkable Ring Signatures.}
The key-image mechanism originates with LSAG \cite{liu2004linkable} and CryptoNote \cite{van2013cryptonote}, refined by MLSAG \cite{noether2016ring} and CLSAG \cite{goodell2019concise} in Monero. Obscura-PQ's two response-sharing relations are the lattice analogue of CLSAG's parallel equations. Among post-quantum constructions, Raptor \cite{lu2019raptor} builds linkable ring signatures from NTRU trapdoors; Lattice RingCT \cite{alberto2018post}, the one-out-of-many framework of Esgin et al. \cite{esgin2019lattice}, and MatRiCT \cite{esgin2019matrict} develop the commitment-and-serial-number approach that Obscura-PQ instantiates; DualRing \cite{yuen2021dualring} gives a compressible challenge-ring, single-response closure with a lattice instantiation; and Calamari and Falafl \cite{beullens2020calamari} achieve logarithmic-size (linkable) ring signatures from isogenies and lattices respectively, with verifiers that are, however, far beyond current on-chain budgets. To our knowledge, none of these schemes has previously been deployed natively inside a deployed smart contract.

\paragraph{On-Chain Verification Under Budget Constraints.}
While empirical deanonymization studies of deployed privacy coins and protocols \cite{moser2018empirical,kumar2017traceability,beres2021blockchain} motivate the client-side decoy discipline detailed in Section~\ref{sec:protocol}, the primary bottleneck for smart-contract privacy remains on-chain verification. Existing zero-knowledge proof systems struggle in this environment: post-quantum succinct systems like STARKs \cite{ben2018scalable} feature transparent setups and sublinear verification, yet their proof sizes and verifier costs are completely impractical under strict per-group opcode limits. Conversely, highly efficient inner-product arguments like Bulletproofs \cite{bunz2018bulletproofs} rely on classical discrete-logarithm assumptions. Obscura-PQ bridges this gap. By leveraging a split-transform phase machine and box-persisted state, it provides, to our knowledge, the first demonstration that complex structured-lattice verification can be systematically decomposed to run natively within highly constrained smart-contract execution environments.

\section{Preliminaries}
\label{sec:prelim}
This section formalizes the algebraic setting, the hardness assumptions, the signature paradigm, and the blockchain execution model that underpin Obscura-PQ.

\subsection{Cryptographic Foundations}

\paragraph{The Ring $\Rq$ and the Negacyclic NTT.}
All arithmetic takes place in the cyclotomic quotient ring
\[
\Rq \;=\; \mathbb{Z}_q[X]/(X^{n}+1), \qquad n = 512,\quad q = 12289,
\]
the ring underlying the Falcon signature scheme \cite{fouque2018falcon}, which NIST selected for post-quantum standardization. Since $2n \mid q-1$, the modulus splits completely and $\Rq$ supports a full negacyclic number-theoretic transform (NTT): a bijective evaluation map $\NTT : \Rq \to \mathbb{Z}_q^{n}$ under which ring multiplication becomes pointwise multiplication, computable with $\tfrac{n}{2}\log_2 n = 2304$ butterfly operations \cite{longa2016speeding}. We write $\hat{x} = \NTT(x)$ and $\odot$ for pointwise multiplication modulo $q$. For $f \in \Rq$ we write $\lVert f \rVert^2$ for the squared Euclidean norm of its integer coefficient vector; \emph{short} elements are those with small norm.

\paragraph{Hardness Assumptions.}
Obscura-PQ relies on two standard structured-lattice assumptions and the random-oracle heuristic:
\begin{itemize}
    \item \textbf{B1 (Ring-SIS)} \cite{lyubashevsky2006generalized,peikert2007lattices}: Given uniformly random $a_1, \dots, a_m \in \Rq$, it is computationally infeasible to find short polynomials $f_1, \dots, f_m$, not all zero, such that $\sum_i a_i f_i = 0$. In Obscura-PQ, the Ring-SIS assumption guarantees three critical properties: \begin{itemize} \item \emph{Commitment Binding:} The coin commitment $C = a_1 k + a_2 s + a_4 e$ is a rank-one module instance over $(a_1, a_2, a_4)$. Opening $C$ to two distinct short triples yields a direct Ring-SIS collision. \item \emph{Unforgeability:} An adversary cannot forge a valid spend proof for an uncorrupted coin without finding a short solution to the underlying linear relations. \item \emph{Serial-Number Soundness:} Because the opening $(k, e)$ is uniquely pinned by commitment binding, the deterministic serial number $\sn = a_3 k + e$ is strictly fixed per coin. A malicious signer cannot equip a single coin with multiple valid serial numbers without computing a Ring-SIS collision. \end{itemize}
    
    \item \textbf{B2 (Ring-LWE)} \cite{lyubashevsky2010ideal}: For a public uniform $a \in \Rq$ and short $x, y$, the sample $(a, a x + y)$ is computationally indistinguishable from uniform. In Obscura-PQ, this assumption provides the \emph{hiding} properties necessary for anonymity and unlinkability, applying in two distinct contexts: \begin{itemize} \item \emph{Commitment Hiding:} The coin commitment $C = a_1 k + a_2 s + a_4 e$ acts as a Ring-LWE sample), making it pseudorandom and computationally hiding its opening. \item \emph{Serial-Number Pseudorandomness:} The serial number $\sn = a_3 k + e$ takes the exact form of a Ring-LWE sample with secret $k$ and short error $e$. Because $e = H_{\mathrm{short}}(k)$ is derived deterministically, $\sn$ is formally a \emph{self-correlated-noise} Ring-LWE instance. By modeling $H_{\mathrm{short}}$ as a random oracle with short outputs, $e$ acts as an independent short error to any adversary that has not queried the oracle at $k$. Consequently, recovering $k$ from $\sn$ reduces to \emph{search} Ring-LWE, and linking a deposit to a withdrawal reduces to \emph{decisional} Ring-LWE in the random-oracle model (see Appendix~\ref{app:formal}, Assumption~\ref{ass:lwe} and Remark~\ref{rem:gaps}). \end{itemize}

    \item \textbf{B3 (Random-Oracle Model)} \cite{bellare1993random}: The hash function instantiating the Fiat--Shamir transform \cite{fiat1986prove} is modeled as a random oracle. While the security reductions presented in this paper operate within the classical random-oracle model (ROM), the intended long-term deployment model requires security in the quantum-accessible ROM (QROM) \cite{boneh2011random,don2019security}. A formal QROM analysis of the composed scheme remains an important direction for future work (see Section~\ref{sec:sec}).
\end{itemize}
Both lattice assumptions are quantum-hard, no quantum algorithm better than generic lattice reduction is known, and they underlie NIST-selected post-quantum signatures, including the standardized ML-DSA (CRYSTALS-Dilithium) and Falcon \cite{ducas2018crystals,fouque2018falcon}. In contrast to SNARK-based pools, no trusted setup is required anywhere in the protocol.

\paragraph{Fiat--Shamir with Aborts.}
In lattice-based sigma protocols, the prover's response has the form $z = y + \chi \cdot w$, where $y$ is a short random mask, $\chi$ a sparse challenge, and $w$ the secret. Published naively, the distribution of $z$ leaks $w$. Lyubashevsky's \emph{rejection sampling} technique \cite{lyubashevsky2009fiat,lyubashevsky2012lattice} restarts the signing attempt unless $z$ falls within an acceptance region, making the accepted distribution (nearly) independent of the secret; verifiers additionally enforce a norm bound $\lVert z \rVert^2 \le \beta^2$, which is essential for soundness, since without it the verification relations are satisfiable trivially.

\paragraph{Linkable Ring Signatures.}
A ring signature \cite{rivest2001leak,abe20041,cramer1994proofs} lets a signer sign on behalf of an ad-hoc set of public keys without revealing which one it controls; unlike group signatures \cite{chaum1991group}, no manager or coordination is needed. \emph{Linkable} ring signatures \cite{liu2004linkable} add a key image (linking tag): a deterministic public function of the signing secret that is identical across all signatures produced with that secret, enabling double-spend detection while preserving signer ambiguity \cite{van2013cryptonote,noether2016ring,goodell2019concise}. Post-quantum linkable ring signatures follow two main paradigms: trapdoor- or preimage-based tags \cite{lu2019raptor}, and the \emph{commitment-and-serial-number} approach of Lattice RingCT and MatRiCT \cite{alberto2018post,esgin2019matrict}, in which the tag's uniqueness is inherited from the binding of a lattice commitment. Obscura-PQ adopts the latter paradigm, combined with the sequential challenge-chain ring closure of AOS \cite{abe20041} and Borromean ring signatures \cite{Maxwell2015BorromeanRS}, and a two-equation response-sharing structure that is the lattice analogue of the parallel equations that bind the key image in CLSAG \cite{goodell2019concise}. We note that DualRing \cite{yuen2021dualring} realizes a structurally different (challenge-ring, single-response) closure that admits logarithmic compression. While we borrow its response-sharing intuition, we explicitly reject its challenge-ring construction in favor of a linear-size AOS-style chain. Although the AOS-style chain yields larger proof sizes, its sequential structure is essential for our design: it allows the heavy lattice arithmetic to be paused, persisted to storage, and resumed, which maps directly onto the phase-decomposed on-chain verifier required by Algorand's strict execution limits (Section~\ref{sec:onchain}). Sublinear post-quantum ring signatures also exist \cite{beullens2020calamari}, at the cost of substantially heavier verifiers; we discuss the trade-off in Section~\ref{sec:related}.

\subsection{Blockchain Execution Model}
Obscura-PQ targets the Algorand Virtual Machine (AVM) \cite{algorand_avm}, whose resource model shapes the entire on-chain design:

\paragraph{Opcode Budget and Inner-Transaction Pooling.}
Every application call executes under a deterministic per-transaction budget of 700 opcode units, pooled across the transactions of an atomic group \cite{algorand_sc_costs_constraints,algorand_tx_fees}. A contract may issue \emph{inner transactions}; each inner application call contributes a further 700 units to the shared pool. By bursting inner calls to a stateless ``dummy'' application (whose single method simply approves), a contract provisions computational headroom dynamically, a technique we refer to as \emph{OpUp pooling}. The pooled budget of a single group is nevertheless capped at roughly $1.9 \times 10^{5}$ units, which is \emph{below} the cost of one forward NTT: post-quantum verification therefore cannot fit in one group and must be decomposed into a sequence of groups with persistent intermediate state (Section~\ref{sec:onchain}).

\paragraph{Box Storage and Minimum Balance.}
Algorand \emph{box storage} \cite{algorand_box_storage} is an unbounded key--value store scoped to an application. Obscura-PQ stores each commitment and each nullifier as an individual box, so ring-membership and double-spend checks are $O(1)$ assertions on box existence, with no Merkle accumulator, no inclusion proofs, and no global state contention. Every box locks a \emph{minimum balance requirement} (MBR) of $2500 + 400 \cdot (\text{key} + \text{value})$ micro-ALGO in the application account \cite{algorand_box_storage}; boxes are also the only mutable storage that survives across transaction groups, which the protocol exploits to persist verification state. Several further consensus limits matter: application arguments are capped at 2\,048 bytes in total, each transaction carries at most 8 box/application/account references, each box reference adds only 1\,024 bytes of I/O budget charged against the \emph{full} size of every touched box, an atomic transaction group holds at most 16 transactions, and, although the number of boxes is unbounded, a single box value may not exceed 32\,768 bytes \cite{algorand_protocol_parameters}. The last limit is the one that ultimately caps the ring size (Section~\ref{sec:eval}), since the transport box grows linearly in $r$; the 16-transaction cap merely dictates that the transport upload be split across several groups at the largest ring sizes.

\section{Cryptographic Construction}
\label{sec:crypto}
This section formalizes the lattice linkable ring signature at the core of Obscura-PQ. The scheme consists of the algorithms $\CoinKeyGen$, $\SerNum$, $\Sig$, $\Ver$, and $\Link$; Appendix~\ref{app:formal} gives the formal syntax and security definitions.

\subsection{Public Parameters and Challenge Space}
\label{sec:params}
Four public ring elements $a_1, a_2, a_3, a_4 \in \Rq$ are derived by a nothing-up-my-sleeve hash-to-ring map,
\[
a_j \;=\; H_{\mathrm{ring}}(\texttt{DOM\_PP} \,\Vert\, j), \qquad j \in \{1,2,3,4\},
\]
where the index $j$ is encoded as four big-endian bytes, $H_{\mathrm{ring}}$ expands SHA-256 \cite{nist2015sha} in counter mode ($\mathrm{SHA\text{-}256}(\mathrm{seed} \Vert \mathrm{counter})$ blocks) and reads successive 16-bit big-endian chunks modulo $q$, and \texttt{DOM\_PP} is the fixed domain-separation string \texttt{OBSCURA-PQ/DOM\_PP/v2}. Anyone can re-derive and audit the parameters; there is no trusted setup.

The challenge space $\mathcal{C}$ consists of \emph{sparse ternary} polynomials: elements of $\Rq$ with exactly $w = 48$ nonzero coefficients, each $\pm 1$. Its size satisfies $\lvert \mathcal{C} \rvert = \binom{512}{48}\cdot 2^{48} > 2^{256}$, so challenges carry at least 256 bits of entropy. A deterministic map $\HtoC : \{0,1\}^{256} \to \mathcal{C}$ expands a 32-byte seed into positions and signs via SHA-256 counter-mode blocks, skipping collisions; sparsity makes multiplication by a challenge a cheap signed-shift accumulation ($O(nw)$ additions, no transform), which both the prover and the on-chain verifier exploit. All remaining protocol parameters are collected in Table~\ref{tab:params}.

\begin{table}[t]
\centering
\caption{Obscura-PQ system parameters.}
\label{tab:params}
\begin{tabularx}{\textwidth}{@{} l @{\hspace{1.5em}} X r @{}}
\toprule
\textbf{Symbol} & \textbf{Meaning} & \textbf{Value} \\
\midrule
$n$, $q$ & ring degree and modulus of $\Rq = \mathbb{Z}_q[X]/(X^n+1)$ & $512$, $12289$ \\
$\eta$ & width of the centred-binomial secret distribution & $2$ \\
$\varsigma$ & standard deviation of the mask distribution & $2500$ \\
$\beta^2$ & squared-norm acceptance bound per response & $4.2 \times 10^{9}$ \\
$w$ & Hamming weight of ternary challenges ($\lvert\mathcal{C}\rvert > 2^{256}$) & $48$ \\
$\mathcal{H}$ & random oracle (challenge chain, $\HtoC$, $H_{\mathrm{ring}}$, nullifier) & SHA-256 \\
$r_{\max}$ & on-chain ring-size cap & $10$ \\
\bottomrule
\end{tabularx}
\end{table}

\subsection{Coin Generation (\texorpdfstring{\CoinKeyGen}{CoinKeyGen}) and Serial Number (\texorpdfstring{\SerNum}{SN})}
\label{sec:coin}
To deposit, a user samples a coin opening and publishes a binding commitment:
\begin{enumerate}
    \item Sample the spend secret $k$ and commitment randomness $s$ independently from the centred binomial distribution with parameter $\eta = 2$ (coefficients in $[-2,2]$), and derive the deterministic short \emph{serial-number blinding} $e = H_{\mathrm{short}}(k)$, a centred-binomial element expanded from $k$ under a domain-separated hash.
    \item Compute the \emph{coin commitment}
    \[
        C \;=\; a_1 k + a_2 s + a_4 e \;\in\; \Rq ,
    \]
    a rank-one Ring-SIS commitment in the style of \cite{baum2018more}, and publish its forward transform $\NTT(C)$ (1\,024 bytes: 512 coefficients of 2 bytes).
\end{enumerate}
Binding follows from B1: opening $C$ to a second short triple $(k', s', e') \neq (k,s,e)$ yields the Ring-SIS collision $a_1(k - k') + a_2(s - s') + a_4(e - e') = 0$. Hiding follows from B2. Each coin uses fresh, independent randomness, so compromise of one coin does not propagate.

The \emph{serial number} (key image) of a coin is the deterministic Ring-LWE image
\[
    \sn \;=\; a_3 k + e \;\in\; \Rq ,
\]
computed by $\SerNum(k)$ at withdrawal time and revealed \emph{in full}, transported as $\NTT(\sn)$. The compact double-spend key is its digest
\[
    \snstar \;=\; \mathrm{SHA\text{-}256}\bigl(\pack(\NTT(\sn))\bigr).
\]
Because the opening $(k, e)$ is unique (binding) and $\sn$ is deterministic in it, every spend of a coin, honest or adversarial, yields the identical $\snstar$; and because $\sn = a_3 k + e$ has the form of a Ring-LWE sample (with $e = H_{\mathrm{short}}(k)$ modeled as a random-oracle short error), recovering $k$ from $\sn$ (search Ring-LWE) and deciding whether a given commitment shares its $k$ with a given serial number (decisional Ring-LWE in the random-oracle model, Remark~\ref{rem:gaps}) are both hard, so revealing $\sn$ neither exposes the secret nor links the withdrawal to a deposit. The blinding $e$ is bound into $C$ via $a_4 e$, so its determinism, hence double-spend soundness, follows from binding rather than from the honest hash derivation.

\subsection{The Two Verification Relations and the Wire Format}
\label{sec:relations}
Fix a ring of commitments $\mathbf{C} = (C_0, \dots, C_{r-1})$ and a serial number $\sn$. For a response triple $(z_{k,i}, z_{s,i}, z_{e,i})$ and challenge $\chi_i \in \mathcal{C}$, define the two \emph{response-sharing} linear relations
\begin{equation}
\label{eq:relations}
    t^{\mathrm{pk}}_i \;=\; a_1 z_{k,i} + a_2 z_{s,i} + a_4 z_{e,i} - \chi_i\, C_i,
    \qquad
    t^{\mathrm{sn}}_i \;=\; a_3 z_{k,i} + z_{e,i} - \chi_i\, \sn .
\end{equation}
The first opens the ring commitment; the second reproduces the serial number. Both consume the \emph{same} responses $z_{k,i}$ and $z_{e,i}$, which is what binds the published serial number to the spent coin: a witness satisfying both relations at one index is a single short $(k, e)$ that simultaneously opens $C_i$ and satisfies $\sn = a_3 k + e$ (Appendix~\ref{app:formal}).

By linearity of the NTT, the relations are evaluated pointwise in the frequency domain (note the identity multiplier on $z_{e,i}$ in the second relation, whose transform is the all-ones vector):
\begin{equation}
\label{eq:freq}
    \widehat{t^{\mathrm{pk}}_i} = \hat{a}_1 \odot \hat{z}_{k,i} + \hat{a}_2 \odot \hat{z}_{s,i} + \hat{a}_4 \odot \hat{z}_{e,i} - \hat{\chi}_i \odot \hat{C}_i,
    \qquad
    \widehat{t^{\mathrm{sn}}_i} = \hat{a}_3 \odot \hat{z}_{k,i} + \hat{z}_{e,i} - \hat{\chi}_i \odot \hat{\sn} .
\end{equation}

This motivates the protocol's wire format: commitments are \emph{stored} as $\NTT(C)$, the serial number is \emph{transported} as $\NTT(\sn)$, the fixed-base transforms $\hat{a}_1, \hat{a}_2, \hat{a}_3, \hat{a}_4$ are precomputed once at deployment, and the Fiat--Shamir chain hashes the packed frequency-domain images. Consequently the verifier computes only \emph{forward} transforms of per-member data ($z_{k,i}$, $z_{s,i}$, $z_{e,i}$, $\chi_i$) and never an inverse transform. Since the NTT is a bijection on $\Rq$, binding is preserved and $\NTT(\sn)$ determines $\sn$ exactly, the serial number remains revealed in full.

\subsection{Signing (\texorpdfstring{\Sig}{Sign})}
\label{sec:sign}
The message is the \emph{settlement context}
\[
    \mu \;=\; \mathrm{recipient}(32) \,\Vert\, \mathrm{relayer}(32) \,\Vert\, \mathrm{fee}(8) \,\Vert\, \mathrm{appID}(8) \quad \text{(80 bytes, big-endian)},
\]
binding all values an adversary could profitably alter. The signer holds the opening $(k,s,e)$ of $C_\pi$ at a secret index $\pi \in [0, r)$ and proceeds as follows; the Fiat--Shamir chain values $c_i \in \{0,1\}^{256}$ seed the ternary challenges $\chi_i = \HtoC(c_i)$, and the chain-update function is
\[
    c_{i+1 \bmod r} \;=\; \mathcal{H}\Bigl(\mu \,\Vert\, \pack(\hat{\sn}) \,\Vert\, \pack\bigl(\widehat{t^{\mathrm{pk}}_i}\bigr) \,\Vert\, \pack\bigl(\widehat{t^{\mathrm{sn}}_i}\bigr) \,\Vert\, (i{+}1 \bmod r) \Bigr),
\]
with the index encoded as 2 bytes.
\begin{enumerate}
    \item \textbf{Commitment at the true index.} Sample short masks $y_k, y_s, y_e$ from the rounded Gaussian $D_\varsigma$ and compute
    \[
        \widehat{t^{\mathrm{pk}}_\pi} = \hat{a}_1 \odot \hat{y}_k + \hat{a}_2 \odot \hat{y}_s + \hat{a}_4 \odot \hat{y}_e,
        \qquad
        \widehat{t^{\mathrm{sn}}_\pi} = \hat{a}_3 \odot \hat{y}_k + \hat{y}_e ,
    \]
    then derive $c_{\pi+1}$ by the chain update.
    \item \textbf{Decoy simulation.} For $i = \pi{+}1, \dots, \pi{-}1$ (indices mod $r$): sample responses $z_{k,i}, z_{s,i}, z_{e,i} \leftarrow D_\varsigma$, set $\chi_i = \HtoC(c_i)$, evaluate the relations~\eqref{eq:freq}, and derive $c_{i+1}$ by the chain update.
    \item \textbf{Closure.} At the true index, set $\chi_\pi = \HtoC(c_\pi)$ and close over the integers:
    \[
        z_{k,\pi} = y_k + \chi_\pi\, k, \qquad z_{s,\pi} = y_s + \chi_\pi\, s, \qquad z_{e,\pi} = y_e + \chi_\pi\, e .
    \]
    \item \textbf{Rejection sampling.} Compute the standard Fiat--Shamir-with-aborts acceptance probability \cite{lyubashevsky2012lattice}, a product of the per-response factors
    \[
        \rho \;=\; \frac{1}{M} \prod_{x \in \{k,s,e\}} \exp\left( \frac{-2\langle z_{x,\pi}, \chi_\pi x \rangle + \lVert \chi_\pi x \rVert^2}{2\varsigma^2} \right)
    \]
    for a bounding constant $M$, and accept with probability $\rho$. If rejected, or if any of $\lVert z_{k,\pi} \rVert^2$, $\lVert z_{s,\pi} \rVert^2$, $\lVert z_{e,\pi} \rVert^2$ exceeds $\beta^2$, restart from step 1 with fresh masks. The probabilistic acceptance ratio decouples the distribution of accepted responses exactly from the secret $(k,s,e)$.
\end{enumerate}
The signature is
\[
    \sigma \;=\; \bigl(c_0,\; \{(z_{k,i},\, z_{s,i},\, z_{e,i})\}_{i=0}^{r-1}\bigr),
\]
of size exactly $32 + 3072r$ bytes (responses are packed as 512 signed 16-bit coefficients each). Sparsity of $\chi_\pi$ and shortness of $(k,s,e)$ keep the closure short, so honest responses pass the bound with overwhelming probability while the bound excludes the trivial (unboundedly long) solutions an adversary could otherwise construct.

\subsection{Verification (\texorpdfstring{\Ver}{Verify}) and Linking (\texorpdfstring{\Link}{Link})}
\label{sec:verify}
Given $(\mathbf{C}, \mu, \hat{\sn}, \sigma)$, the verifier sets $c \leftarrow c_0$ and, for $i = 0, \dots, r-1$:
\begin{enumerate}
    \item asserts the norm bounds $\lVert z_{k,i} \rVert^2 \le \beta^2$, $\lVert z_{s,i} \rVert^2 \le \beta^2$, and $\lVert z_{e,i} \rVert^2 \le \beta^2$ on the \emph{raw signed} coefficient vectors;
    \item computes $\chi_i = \HtoC(c)$ and the frequency-domain images~\eqref{eq:freq};
    \item updates $c$ by the chain-update function.
\end{enumerate}
It accepts iff the final $c$ equals $c_0$ (ring closure). $\Link$ declares two withdrawals linked iff their nullifiers $\snstar$ coincide. Verification performs, per member, exactly four forward NTTs (for $z_{k,i}$, $z_{s,i}$, $z_{e,i}$, and $\chi_i$), one pointwise combine, three norm checks, and one hash. This cost profile drives the on-chain design detailed in Section~\ref{sec:onchain}. Norm checks are performed on the raw signed integers as transmitted; checking a reduced representative would wrap around and be unsound.

\section{Protocol Description}
\label{sec:protocol}
Obscura-PQ operates in two phases, deposit and withdrawal, mediated by a fixed-denomination pool contract. Throughout, $\mathcal{B}_C$ and $\mathcal{B}_N$ denote the on-chain commitment and nullifier box sets.

\subsection{Deposit Phase}
As illustrated in Figure~\ref{fig:seq_deposit}, the user runs $\CoinKeyGen$ locally, stores the opening $(k,s,e)$, and submits an atomic group of two transactions: an application call \textsf{deposit} carrying $\NTT(C)$, and a payment of the fixed denomination \textsc{denom} to the application address. The contract asserts the group shape, the payment amount and receiver, and the payload length; asserts that the commitment box keyed by
\[
    \texttt{"c"} \,\Vert\, \mathrm{SHA\text{-}256}\bigl(\NTT(C)\bigr)
\]
does not yet exist (duplicate commitments are rejected); creates the box with value $\NTT(C)$; and increments a deposit counter. The 1\,024-byte box locks an MBR of 425{,}300 micro-ALGO in the application account, the protocol's principal per-deposit storage cost, which is deducted from the eventual payout (Section~\ref{sec:eval}). The serial number is \emph{not} published at deposit; only the commitment is.

\subsection{Anonymity-Set Construction}
Prior to withdrawal the client forms a ring $\mathbf{C} = (C_0, \dots, C_{r-1})$ of on-chain commitments containing its own coin at a secret index, entirely off-chain:
\begin{enumerate}
    \item Fetch the commitment boxes and verify the user's own commitment is present;
    \item Query the indexer for recent deposits and order the decoy pool by recency; decoy sampling that fails to mimic plausible spending behavior enables statistical deanonymization \cite{moser2018empirical,kumar2017traceability}, and recency bias counters temporal intersection attacks;
    \item Draw decoys uniformly from a \emph{bounded recent pool} (the 20 most recent valid decoys), avoiding the deterministic ``always the newest $r{-}1$'' selection;
    \item Randomly select $r-1$ decoy commitments from the recent pool (up to the maximum ring size $r_{\max}$), insert the user's own commitment, and apply a Fisher--Yates shuffle \cite{fisher1953statistical} seeded from a cryptographically secure generator, making the signer's position uniform.
\end{enumerate}

\begin{figure}[t]
\centering
\includegraphics[width=0.88\textwidth]{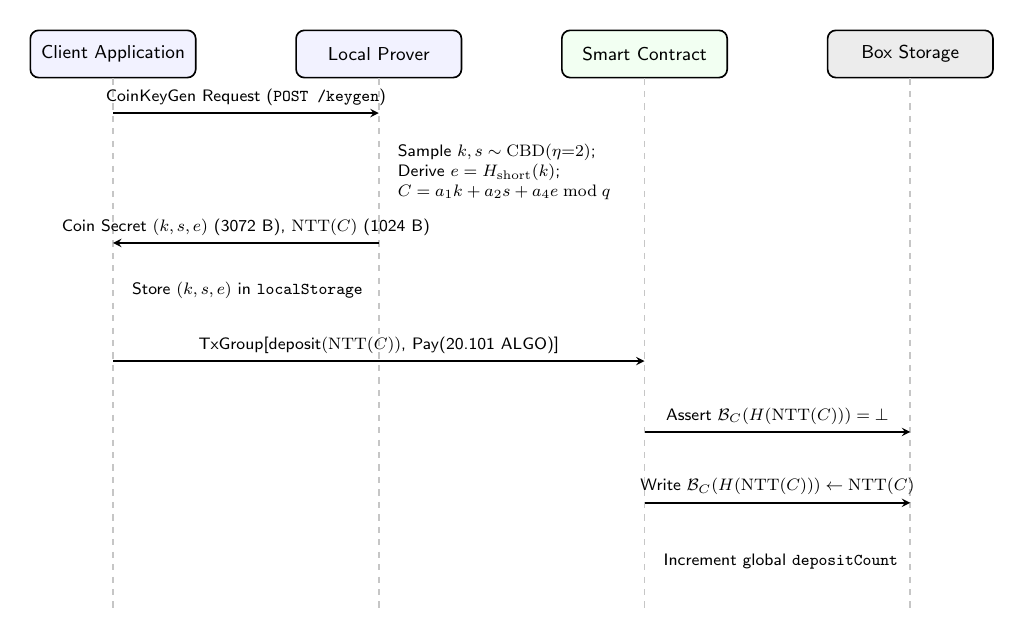}
\caption{Deposit phase sequence diagram. The client application delegates \textsf{CoinKeyGen} to the local prover, which samples the short coin secret $(k, s)$ from a centred binomial distribution ($\eta = 2$), derives the deterministic blinding $e = H_{\mathrm{short}}(k)$, and computes the Ring-SIS commitment $C = a_1 k + a_2 s + a_4 e \bmod q$ over the cyclotomic ring. The client stores the returned coin secret locally and submits an atomic transaction group pairing the \textsf{deposit} application call with the fixed-denomination payment of 20.101 ALGO. The smart contract verifies the payment, rejects duplicate commitments, allocates a new commitment box keyed by the hash of $\mathrm{NTT}(C)$ in on-chain storage, and increments the global deposit counter.}
\label{fig:seq_deposit}
\end{figure}

\subsection{Withdrawal Phase}
\label{sec:withdrawal}
As depicted in Figure~\ref{fig:seq_withdraw}, the withdrawal phase begins off-chain. The client assembles the context $\mu$ (recipient, relayer, total fee, application ID) and requests a proof from the local prover, receiving $\snstar$, $\NTT(\sn)$, and $\sigma$. In the implementation, the withdrawing wallet \emph{self-relays}: it fronts all transaction fees, and the total is bound into $\mu$ as the fee and reimbursed from the pool at settlement. The protocol equally supports third-party relayers; because the relayer address and fee are signature-bound, a relayer can neither redirect funds nor overcharge. However, a relayer network is not part of the current implementation.

On-chain, the withdrawal spans a sequence of transaction groups (all constructed and signed in a single wallet interaction), detailed in Section~\ref{sec:onchain}: the proof is streamed into a transport box, $9r$ verification phases recompute the challenge chain member by member, and the final \textsf{withdraw} call checks closure and settles. Settlement is atomic: the contract deletes the transient boxes, records $\snstar$ in $\mathcal{B}_N$, and issues three inner payments, the payout $\textsc{denom} - \mathrm{fee} - c_{\mathrm{stor}}$ to the recipient (where $c_{\mathrm{stor}} = 441{,}000$ micro-ALGO covers the permanent commitment and nullifier MBRs), the bound fee to the relayer, and the refund of all fronted transient MBRs.

\begin{figure}[htbp]
\centering
\includegraphics[width=\textwidth]{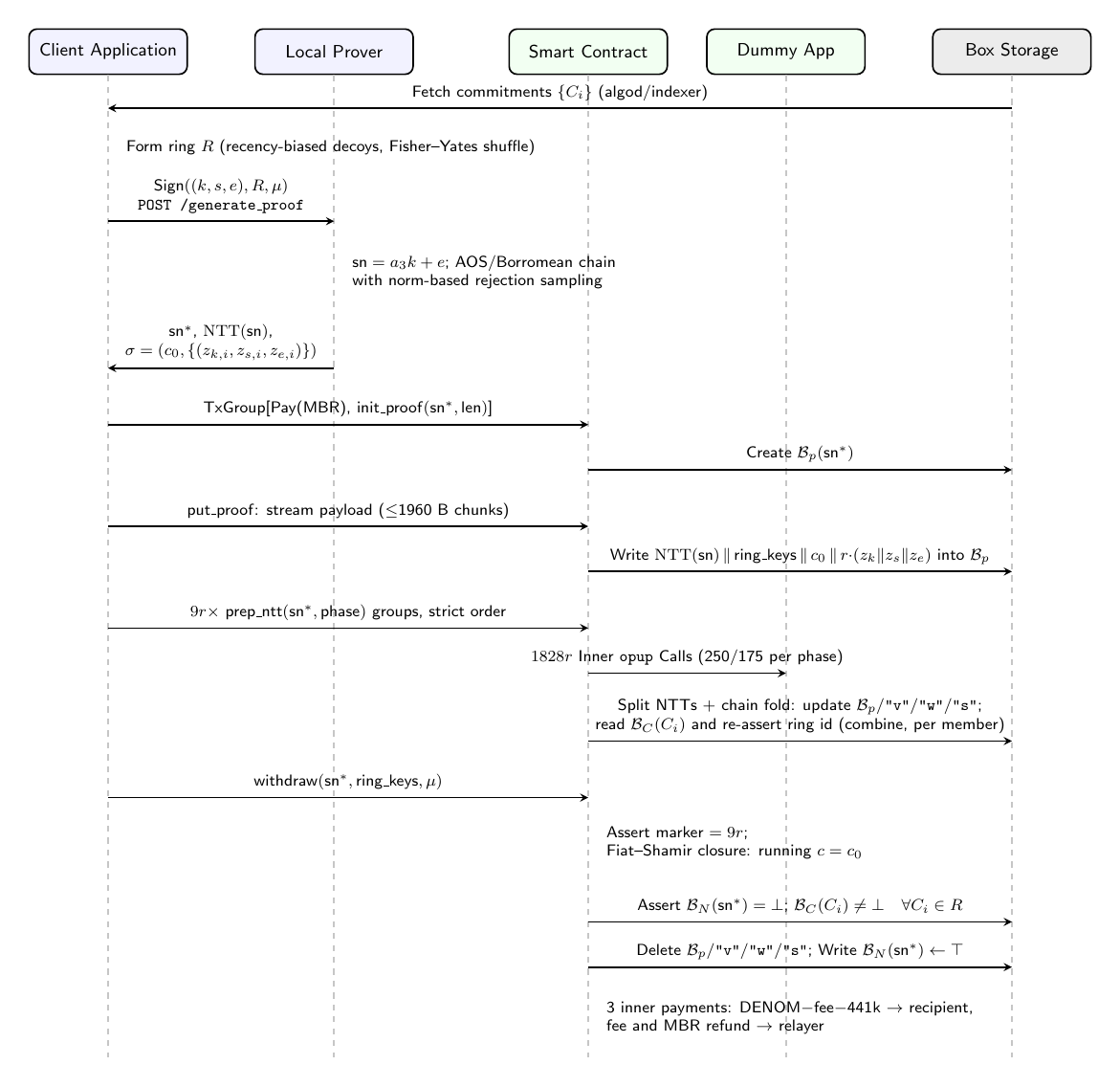}
\caption{Withdrawal phase sequence diagram. The client application constructs an anonymity set from on-chain public commitments and delegates signing to the local prover, which derives the serial number $\mathsf{sn} = a_3 k + e$ and produces the AOS/Borromean-style lattice ring signature under norm-based rejection sampling. Because the signature exceeds application-argument limits, the client first creates a proof transport box (\textsf{init\_proof}) and streams the payload into it in chunks (\textsf{put\_proof}). Verification then runs as $9r$ strictly ordered \textsf{prep\_ntt} groups that perform the split NTTs and fold the Fiat--Shamir challenge chain, expanding the opcode budget via $1828r$ inner \texttt{opup} calls to a dummy application and re-asserting each ring member's commitment against Box Storage. The final \textsf{withdraw} call asserts chain closure ($c = c_0$), validates the nullifier and ring membership, deletes the transient boxes, records the nullifier, and settles the payment with three inner transactions.}
\label{fig:seq_withdraw}
\end{figure}

\subsection{Double-Spend Prevention}
The digest $\snstar$ is the nullifier. The contract binds it to the transported serial number by recomputing $\snstar = \mathrm{SHA\text{-}256}(\pack(\NTT(\sn)))$, asserts that no box $\texttt{"n"} \Vert \snstar$ exists, and creates that box atomically with the payout. Because the opening $(k, e)$ is unique by commitment binding and $\sn = a_3 k + e$ is deterministic, any second spend of the same coin presents the identical $\snstar$ and reverts, regardless of the decoy set chosen. A malicious signer cannot substitute a fresh serial number: a verifying signature with $\sn' \neq a_3 k + e$ would require opening $C_\pi$ to a second short opening (a Ring-SIS collision) or violating relation~\eqref{eq:relations} (Appendix~\ref{app:formal}).

\section{On-Chain Execution}
\label{sec:onchain}
This section describes how the verifier of Section~\ref{sec:verify} executes natively on the AVM. The main difficulty is quantitative: verifying \emph{one ring member} requires four degree-512 forward NTTs, and a single forward NTT with its norm check consumes approximately $2.05 \times 10^{5}$ opcode units (measured on the deployed contract via the node's \texttt{simulate} endpoint), whereas one transaction group pools at most about $1.9 \times 10^{5}$ units. Lattice verification therefore cannot fit in any single group, and the design decomposes it into a resumable sequence of groups with box-persisted intermediate state.

\subsection{Proof Transport via Box Streaming}
The proof material, $\NTT(\sn)$ (1\,024 B), the $r$ ring-member identifiers (32 B each), the chain seed $c_0$ (32 B), and the responses ($3072r$ B), totals $1056 + 3104r$ bytes and exceeds the 2\,048-byte application-argument limit for any $r \ge 1$. The client therefore streams it into an application-owned \emph{transport box} keyed $\texttt{"p"} \Vert \snstar$ before verification begins:
\begin{enumerate}
    \item \textsf{init\_proof}: an atomic group in which a payment covering the transport box's MBR plus the transient-state MBR (482{,}600 micro-ALGO) \emph{precedes} the application call that creates the box; the funder's address is recorded in the box header, and only the funder may write to or clear the box.
    \item \textsf{put\_proof}: the payload is written in chunks of at most 1\,960 bytes (the residual application-argument budget), bounds-checked against the declared length. Because each write is offset-addressed and mutually independent (the box becomes load-bearing only when \textsf{prep\_ntt} first reads it), the chunks need not share one atomic group; when their number would exceed the 16-transaction group cap, they are split across several sequential groups, each carrying enough box references to cover the transport box's full-size I/O budget. Writing is frozen permanently once verification starts.
    \item \textsf{clear\_proof}: at any time before settlement the funder may delete the flow boxes and reclaim the full fronted MBR. Box keys are deterministic in $\snstar$, so a withdrawal that failed mid-flow leaves recognizable state; the client detects and clears it automatically before retrying (self-healing).
\end{enumerate}
All transient MBRs are refunded at settlement, so the permanent state growth of a withdrawal is the empty nullifier box alone.

\subsection{Split NTTs and the Per-Member Phase Machine}
\label{sec:phases}
Each forward NTT is split at butterfly length 32 into two opcode-pooled phases: a \emph{heavy} half (norm check and signed-to-modular reduction of the input, the $\psi$-premultiplication fused with the bit-reversal permutation, and butterfly stages up to length 32; 250 inner OpUp calls) and a \emph{light} half (stages 64 through 512; 175 inner OpUps). Intermediate buffers persist in box storage between phases: response transforms overwrite their slots in the transport box in place, and the challenge transform lives in a dedicated box $\texttt{"w"} \Vert \snstar$. All coefficient products are bounded by $q^2 \approx 1.5 \times 10^{8} < 2^{64}$, so native 64-bit arithmetic suffices; the twiddle-factor, $\psi$-power, and bit-reversal tables are compile-time constants generated from a reference implementation (Section~\ref{sec:impl}).

Verification of member $i$ comprises nine phases, invoked as successive \textsf{prep\_ntt} application calls (Figure~\ref{fig:pipeline}): phases 0--1 transform $z_{k,i}$, phases 2--3 transform $z_{s,i}$, phases 4--5 transform $z_{e,i}$, and phases 6--7 derive $\chi_i = \HtoC(c)$ from the running chain value and transform it. Finally, phase 8 performs the pointwise two-relation combine~\eqref{eq:freq}. It reads the fixed-base transforms from the parameter box and the member's commitment box, re-asserts that the claimed ring identifier equals the SHA-256 digest of the stored commitment, and folds the chain update into a running-state box $\texttt{"s"} \Vert \snstar$. Phase 0 of member 0 additionally seeds the state box with $c_0$ and the settlement context (recipient, relayer, fee). A marker box $\texttt{"v"} \Vert \snstar$ enforces strict phase ordering; because the challenge of member $i$ depends on member $i{-}1$'s relation images, members are processed strictly in order. The per-member cost is
\[
    4 \times (250 + 175) + 128 \;=\; 1828 \ \text{inner OpUp calls} \;\approx\; 1.828\ \text{ALGO}
\]
at the 1\,000 micro-ALGO minimum fee, which, together with the growing transport-box MBR, is the practical force keeping rings small (Section~\ref{sec:eval}). Auxiliary \texttt{opup} calls in each group act purely as reference carriers, working around the 8-reference and box-I/O-budget limits.

\begin{figure}[t]
\centering
\includegraphics[width=1\textwidth]{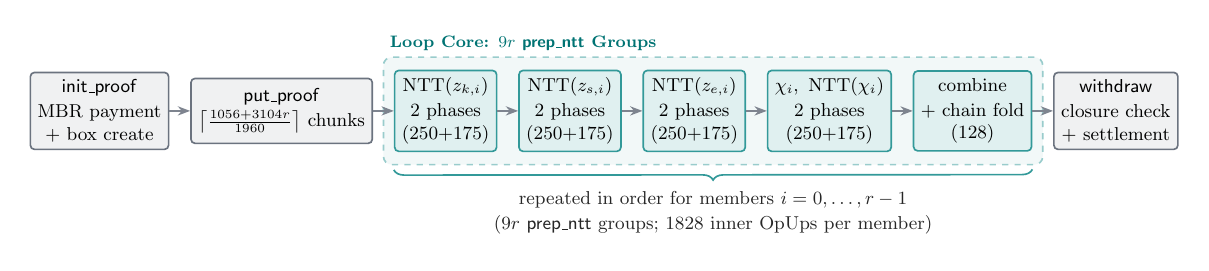}
\caption{The withdrawal pipeline as a sequence of atomic transaction groups. Parenthesized numbers are the inner OpUp calls provisioning each phase's pooled opcode budget. Intermediate transforms persist in box storage between groups; a marker box enforces strict phase ordering, and the final \textsf{withdraw} call asserts that the running Fiat--Shamir chain closed back to $c_0$ before settling. The loop core consists of $9r$ \textsf{prep\_ntt} groups.}
\label{fig:pipeline}
\end{figure}

\subsection{Settlement and Method Semantics}
After all $9r$ phases, the \textsf{withdraw} call asserts: the transport-box layout and the equality of the submitted ring identifiers with the box copy; the binding $\snstar = \mathrm{SHA\text{-}256}(\pack(\NTT(\sn)))$; the marker value $9r$ (all phases complete); \emph{ring closure}, i.e., the running chain value equals $c_0$; the absence of the nullifier box; the existence and digest-match of every ring member's commitment box; and $\mathrm{fee} + c_{\mathrm{stor}} < \textsc{denom}$. It then deletes the flow boxes, creates the nullifier box, and issues the three inner payments of Section~\ref{sec:withdrawal}. Table~\ref{tab:methods} summarizes the contract interface.

Every assertion failure reverts its group, and settlement values (recipient, relayer, fee) are read from the state box, the same values every challenge was computed over, so the payout cannot deviate from what the signature authorized. The nullifier is recorded iff the chain closed, and funds move iff the nullifier is recorded. The groups of the pipeline are \emph{not} atomic with one another; a mid-flow abort strands only refundable boxes, which the self-healing path reclaims. Phase submission is permissionless: any account may advance the pipeline for a given $\snstar$. This never endangers funds, because settlement pays only to the settlement context folded into the challenge chain and any mismatched context makes closure fail; it does mean that a pipeline seeded with an incorrect context, for instance by a party that races the first phase, will not close and must be cleared and restarted by its funder. The design therefore guarantees fund safety but not censorship-resistant liveness for a targeted withdrawal.

\begin{table}[t]
\centering
\caption{Smart-contract methods (routing on the first application argument).}
\label{tab:methods}
\setlength{\tabcolsep}{10pt}
\begin{tabularx}{\textwidth}{@{}l X@{}}
\toprule
\textbf{Method} & \textbf{Guard and effect} \\
\midrule
\textsf{deposit} & group of 2 with \textsc{denom} payment; commitment box absent $\Rightarrow$ create $\texttt{"c"}\Vert H(\NTT(C))$, increment counter \\
\textsf{init\_params} / \textsf{put\_params} & creator-only, deploy-time; stream the 4\,096-byte fixed-base transforms $\hat{a}_1 \Vert \hat{a}_2 \Vert \hat{a}_3 \Vert \hat{a}_4$ into the \texttt{pp} box \\
\textsf{init\_proof} & preceding payment $\ge$ transport MBR $+$ 482{,}600; payload length $= 1056 + 3104r$, $r \in [1, r_{\max}]$; box absent $\Rightarrow$ create $\texttt{"p"}\Vert\snstar$, record funder \\
\textsf{put\_proof} & sender is the recorded funder; verification not started; bounds-checked chunk write \\
\textsf{clear\_proof} & sender is the recorded funder $\Rightarrow$ delete flow boxes, refund the full MBR \\
\textsf{prep\_ntt} & marker equals the requested phase, phase $< 9r$ $\Rightarrow$ run the phase's substep, issue its inner OpUps, advance the marker \\
\textsf{withdraw} & all settlement assertions above $\Rightarrow$ delete flow boxes, create $\texttt{"n"}\Vert\snstar$, three inner payments \\
\textsf{opup} & unconditional approve (budget target / reference carrier) \\
\bottomrule
\end{tabularx}
\end{table}

\section{Implementation}
\label{sec:impl}
Obscura-PQ is deployed as a full-stack application partitioned into three trust domains (Figure~\ref{fig:impl_arch}); the complete implementation source code is publicly available (see \emph{Code Availability}).

\paragraph{Smart-Contract Layer.}
The contract is written in PyTeal \cite{algorand_pyteal} (version 0.27.0) and compiles to TEAL v10 bytecode. All lattice arithmetic, the table-driven negacyclic NTT, signed-to-modular reduction, norm accumulation, the sparse ternary challenge expansion, and the pointwise combine, is implemented directly in TEAL subroutines over byte-packed coefficient arrays. The 4\,096-byte fixed-base transforms would push the approval program past its 8\,192-byte (four-page) size cap and exceed the 2\,048-byte argument limit, so they are provisioned once at deployment into a dedicated parameter box and read only by the combine phase.

\paragraph{Local Prover.}
$\CoinKeyGen$, $\SerNum$, and $\Sig$ execute in a local Python prover exposing an HTTP API on \texttt{localhost}. The prover must run on the user's own machine: the coin opening $(k,s,e)$ is the spend authority, and any party obtaining it can steal the coin and break anonymity. In the current implementation the secret travels between browser and prover over local HTTP and is persisted in plaintext browser storage with manual export/import; a hardened deployment would use an in-process (WebAssembly) prover or authenticated local IPC and an encrypted secret store. Signing takes milliseconds to seconds, dominated by the $O(r)$ ring pass and a small expected number of rejection-sampling restarts.

\paragraph{Client Application.}
The client is a React/TypeScript application handling wallet integration \cite{pera_connect}, ring formation (recency-biased decoy selection over a bounded recent pool, Fisher--Yates shuffling seeded by \texttt{crypto.getRandomValues()} \cite{mdn_crypto_getrandomvalues}), fee and box-reference planning across the pipeline's transaction groups, proof streaming, phase submission, and self-healing cleanup of stranded flows. The repository additionally ships deployment scripts (which provisions the dummy application, funds the contract, and uploads the parameter box), a bytecode-verification script that recompiles the PyTeal source and compares it byte-for-byte against the deployed program, and two analysis tools (a transaction-graph explorer and an indexer-backed transaction classifier) used to audit ledger state.

\begin{figure}[t]
\centering
\includegraphics[width=0.95\textwidth]{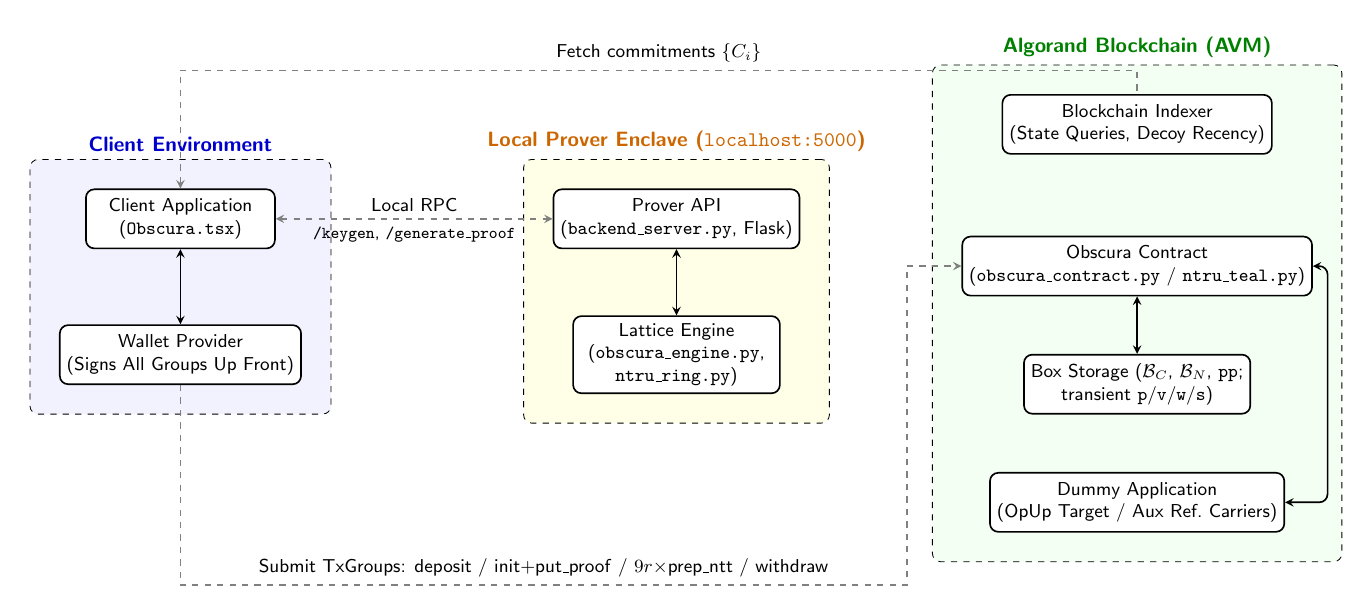}
\caption{Obscura-PQ implementation architecture. The system is partitioned into three distinct trust domains. The client environment handles decoy selection and transaction construction, with the wallet provider signing all withdrawal groups in a single prompt. The local prover enclave, a Flask service on \texttt{localhost:5000}, executes the heavy lattice arithmetic (coin sampling, Ring-SIS commitment, and AOS/Borromean-style signing with rejection sampling) without exposing the coin secret $(k,s,e)$ to the network. The Algorand blockchain enforces state transitions, utilizing the indexer for state queries and decoy recency ordering, Box Storage for membership and nullifier checks alongside the public-parameter and transient withdrawal-flow boxes, and a dummy application for opcode pooling and auxiliary box-reference carrying.}
\label{fig:impl_arch}
\end{figure}

\section{Security and Privacy Analysis}
\label{sec:sec}
We analyze Obscura-PQ under an explicit adversarial model. Formal definitions, theorem statements, and proofs are given in Appendix~\ref{app:formal}; this section states the guarantees, their assumptions, and their limitations.

\subsection{Threat Model}
We assume a public, immutable ledger: a \emph{globally passive adversary} \cite{dingledine2004tor} observes every commitment, every ring, every serial number, every context $\mu$, and every signature. An \emph{active} adversary may submit malformed signatures, attempt double-spends and replays, fabricate ring members, choose serial numbers adversarially, act as a malicious relayer or block proposer, and tamper with transport boxes mid-upload. The \emph{quantum} adversary may run Shor's and Grover's algorithms \cite{shor1994algorithms,grover1996fast}. We assume the Algorand consensus layer and AVM semantics are correct \cite{gilad2017algorand}, and that $\Sig$ executes in a trusted local environment; a remote prover would trivially break both anonymity and theft resistance. Network-layer deanonymization and endpoint compromise are out of scope.

\subsection{Guarantees}
\label{sec:guarantees}

\paragraph{Unforgeability and Theft Resistance.}
The challenge chain is \emph{relaxed} special-sound: two accepting transcripts that share relation images but differ in one challenge yield, by the standard extractor, a \emph{relaxed} short opening $(\bar z_k,\bar z_s,\bar z_e)$ satisfying $a_1 \bar z_k + a_2 \bar z_s + a_4 \bar z_e = \bar\chi\, C_i$ and $a_3 \bar z_k + \bar z_e = \bar\chi\,\sn$ for the challenge difference $\bar\chi$ (Theorem~\ref{thm:unforgeability}; the extractor recovers an exact opening only relative to $\bar\chi$, as short challenge differences are not invertible in the fully splitting ring). Producing a verifying withdrawal for a ring of honest deposits without knowing a ring member's opening therefore breaks Ring-SIS. Two implementation details are essential to soundness and are enforced on-chain: the norm bounds on every raw response (without which the relations are trivially satisfiable), and the ring-existence check with digest matching (without which an adversary could inject a fabricated commitment whose opening it knows). The transport-box freeze prevents rewriting a response after its norm has been checked.

\paragraph{Signer Ambiguity and Unlinkability.}
After rejection sampling, the true-index responses are distributed (up to negligible statistical distance) like the forward-sampled decoy responses and independently of $(k,s,e)$; every challenge is a random-oracle output independent of the signer index; and the commitment $C = a_1 k + a_2 s + a_4 e$ is pseudorandom under Ring-LWE, so the transcript is simulatable from $(\mathbf{C}, \mu, \sn)$ alone. The only residual signer-dependent value is the serial number $\sn = a_3 k + e$, which has the form of a Ring-LWE sample with the deterministic error $e = H_{\mathrm{short}}(k)$; modelling $H_{\mathrm{short}}$ as a random oracle it hides $k$, so its linkage to a commitment is hard under decisional Ring-LWE in the random-oracle model (Remark~\ref{rem:gaps}). Consequently no probabilistic polynomial-time adversary identifies the signer with probability non-negligibly exceeding $1/r$ (Theorem~\ref{thm:anonymity}), and no observer links a deposit to a withdrawal beyond the $1/r$ ring ambiguity. Deciding whether a commitment and a serial number share the same $k$ reduces to decisional Ring-LWE (in the ROM), and recovering the short $k$ from $\sn$ is the search Ring-LWE problem. The anonymity set is the \emph{ring}, of size at most $r_{\max}$, not the whole pool (see Section~\ref{sec:limitations}).

\paragraph{Linkability, Double-Spend Prevention, and Non-Frameability.}
Every verifying signature carries the serial number $a_3 k + e$ of the \emph{unique} short opening of one ring member (Theorem~\ref{thm:linkability}); uniqueness is commitment binding, not an honesty assumption, so double-spend detection is sound even against a malicious signer who deliberately deviates from the signing algorithm. Conversely, to pre-spend or block an honest user's coin an adversary would have to produce a verifying signature carrying that coin's serial number, which by the extraction argument requires the coin's own opening (Lemma~\ref{lem:nonframe}). The contract enforces nullifier uniqueness atomically with the payout, so each coin is spendable exactly once.

\paragraph{Replay and Front-Running Resistance.}
The full context $\mu$ (recipient, relayer, fee, and application identifier) is hashed into every challenge; altering any component breaks ring closure (Lemma~\ref{lem:replay}). The contract additionally settles to the values stored in the chain-bound state box, so a malicious relayer or proposer cannot redirect the payout or inflate the fee even at the settlement step. Replaying an executed withdrawal is blocked by the recorded nullifier, and the application-identifier binding blocks cross-deployment replay.

\paragraph{Auditability.}
As in accountable-privacy designs \cite{garman2016accountable}, a user may voluntarily disclose a coin opening $(k,s,e)$ to an auditor, who can verify the deterministic blinding $e = H_{\mathrm{short}}(k)$ and can then verify both the deposit commitment $C = a_1 k + a_2 s + a_4 e$ and the withdrawal serial number $\sn = a_3 k + e$, establishing provenance without affecting any other participant's anonymity.

\subsection{Limitations and Open Gaps}
\label{sec:limitations}
We state explicitly what the analysis does \emph{not} establish.
\begin{itemize}
    \item \textbf{Ring-bounded anonymity.} Signer ambiguity is $1/r$, not pool-wide hiding as in accumulator-based SNARK pools. Operational behavior (immediate withdrawal after deposit, correlated timing, small pools) can degrade the effective anonymity below the cryptographic ceiling \cite{moser2018empirical,kumar2017traceability,beres2021blockchain}; the client's recency-biased selection and secure shuffling mitigate but cannot eliminate this.
    \item \textbf{Anonymity in the ROM and self-correlated serial-number noise.} Anonymity and unlinkability reduce to decisional Ring-LWE \emph{in the random-oracle model} (ROM). Because the serial-number error $e = H_{\mathrm{short}}(k)$ is a deterministic function of the secret $k$ (bound into the commitment via $a_4 e$), the serial number $\sn = a_3 k + e$ constitutes a \emph{self-correlated-noise} Ring-LWE instance. Consequently, the reduction models $H_{\mathrm{short}}$ as a random oracle rather than claiming a black-box reduction to standard decisional Ring-LWE (see Assumption~\ref{ass:link} and Remark~\ref{rem:gaps}); recovering $k$ from $\sn$ remains search Ring-LWE. A standard-model reduction could be achieved by sampling $e$ independently and storing it as part of the coin, though this would increase the size of the persisted secret. Separately, because short challenge differences are not guaranteed to be invertible in the fully splitting Falcon-512 ring, our unforgeability extractor recovers a \emph{relaxed} opening. This extraction is nevertheless sufficient to prove theft resistance and non-frameability (Appendix~\ref{app:formal}, Remark~\ref{rem:gaps}).
    \item \textbf{Classical ROM proofs.} The reductions in Appendix~\ref{app:formal} are in the classical random-oracle model, following the analysis templates of \cite{lyubashevsky2012lattice,esgin2019lattice,yuen2021dualring}. A dedicated QROM analysis \cite{boneh2011random,don2019security} of the composed scheme has not been carried out.
    \item \textbf{Sampler simplification.} The implemented mask sampler uses a rounded Gaussian with a norm-bound-only acceptance test, rather than a constant-time discrete Gaussian equipped with the full Fiat--Shamir-with-aborts acceptance ratio \cite{lyubashevsky2012lattice}. Furthermore, while true-index responses are conditioned on the norm bound via rejection sampling, decoy responses are drawn from the unconditioned mask distribution. Because the formal zero-knowledge simulation argument assumes the full acceptance ratio and identically distributed responses across all indices, the theoretical anonymity guarantee holds exactly only for the idealized sampler. A production deployment must implement the fully hardened sampler to achieve the stated bounds.
    \item \textbf{Liveness under permissionless phase submission.} Because the \textsf{prep\_ntt} method imposes no sender checks (Section~\ref{sec:onchain}), any account can advance the withdrawal pipeline. While this permissionless design cannot lead to fund theft or redirection (since settlement is strictly bound to the context folded into the challenge chain, and any mismatch causes verification to fail) it does introduce a liveness vulnerability. A malicious third party could race the first phase or squat the transport box, forcing a targeted withdrawal to abort and requiring the user to restart the process. Consequently, the current design prioritizes fund safety over censorship-resistant liveness. Restricting phase submission and box creation exclusively to the proof funder would eliminate this denial-of-service vector.
    \item \textbf{Concrete parameters.} While the ring parameters $(n, q)$ are inherited directly from Falcon-512, the remaining parameters ($\eta$, $\varsigma$, $\beta$, $w$) were selected primarily to ensure correctness and efficient rejection-sampling termination. Because these values have not yet been calibrated against a formal lattice estimator \cite{albrecht2015concrete}, we do not claim a specific bit-security level for the current parameter set.
    \item \textbf{Unaudited implementation.} The TEAL transcription of the lattice arithmetic is validated through extensive mirror testing and live on-chain execution. However, it has not yet been subjected to formal verification or a third-party security audit.
\end{itemize}

\section{Cost Analysis}
\label{sec:eval}
AVM opcode pricing, transaction fees, and storage deposits are deterministic: they depend strictly on protocol constants and the ring size $r$, rather than network congestion. This section develops the resulting closed-form cost model and evaluates the protocol's performance on the Algorand testnet. We explicitly distinguish between exact theoretical costs and empirical measurements. Specifically, the proof transport size (Figure~\ref{fig:proof_transport_size}) and the storage MBRs (Figure~\ref{fig:fees}) are derived in closed form from the protocol constants. In contrast, the remaining metrics including inner \texttt{opup} calls (Figure~\ref{fig:inner_tx}), transaction fees (Figure~\ref{fig:tx_fee}), total wall-clock runtime (Figure~\ref{fig:runtime_total}), local proof generation time (Figure~\ref{fig:proof_time}), local proof size (Figure~\ref{fig:proof_size}), and the detailed latency breakdowns for deposits (Figure~\ref{fig:deposit_breakdown}) and withdrawals (Figure~\ref{fig:withdrawal_breakdown}) are empirical results drawn directly from live executions on the Algorand testnet.

\paragraph{Communication.}
The signature is $32 + 3072r$ bytes and the full transport payload (serial-number transform, ring identifiers, chain seed, responses) is $1056 + 3104r$ bytes, streamed in at most $\lceil (1056 + 3104r)/1960 \rceil$ chunks (Figure~\ref{fig:proof_transport_size}). As detailed in Figure~\ref{fig:proof_size}, the raw proof size scales deterministically with the anonymity set, reaching over 30\,KB at the maximum ring size. This linear growth in $r$ is intrinsic to the AOS/Borromean challenge-chain structure, as in classical LSAG; what changes in the lattice setting is the constant, 3\,104 bytes per member versus 96 for an elliptic-curve instantiation \cite{azimi2026obscura}, which is why box streaming, rather than argument passing, is unavoidable.

\paragraph{Execution.}
Verification costs $9r$ phase groups and $1828r$ inner OpUp calls. As illustrated in Figure~\ref{fig:inner_tx}, this dynamic budget provisioning leads to an $O(r)$ scaling in the total number of inner transactions. Consequently, this translates to about $1.828r$ ALGO of inner fees, plus a fixed overhead of transport and settlement transactions (Figure~\ref{fig:fees}). Furthermore, Figure~\ref{fig:tx_fee} contrasts this linearly scaling withdrawal fee against the constant 0.001 ALGO deposit fee, highlighting the asymmetric cost profile of the protocol. The denomination is derived in the contract from the worst case at $r = r_{\max}$:
\[
    \begin{aligned}
    \textsc{denom} &= \bigl(60 + 1860\,r_{\max}\bigr)\cdot 10^{3}
    + 425{,}300 + 15{,}700 + 10^{6} \\
    &= 20{,}101{,}000\ \text{micro-ALGO}.
    \end{aligned}
\]
where $1860$ over-approximates the per-member OpUps plus phase and reference-carrier transactions, $60$ covers the fixed overhead, the two MBR terms cover the permanent boxes, and $10^{6}$ is slack. The recipient receives $\textsc{denom} - \mathrm{fee} - 441{,}000$ micro-ALGO; the fee is the withdrawer's actual fronted total, bound into $\mu$ and reimbursed to the (self-)relayer.

\paragraph{Storage.}
Table~\ref{tab:storage} itemizes the box MBRs. Permanent state grows by one 1\,024-byte commitment box per deposit ($\approx 0.425$ ALGO locked) and one empty nullifier box per withdrawal ($\approx 0.016$ ALGO); all verification-flow boxes, including the large transport box ($\approx 8.77$ ALGO of MBR at $r = 10$), are deleted at settlement and their deposits refunded. Compared with accumulator-based designs, the flat state model trades Merkle-update computation and state contention for $O(N)$ box storage with a per-deposit constant.

\begin{table}[t]
\centering
\caption{Box storage and minimum-balance accounting (micro-ALGO; MBR $= 2500 + 400(\text{key}+\text{value})$).}
\label{tab:storage}
\begin{tabularx}{\textwidth}{@{}l l l X@{}}
\toprule
\textbf{Box} & \textbf{Value size (B)} & \textbf{MBR} & \textbf{Lifetime} \\
\midrule
commitment $\texttt{"c"}\Vert H(\NTT(C))$ & 1024 & 425{,}300 & permanent, per deposit \\
nullifier $\texttt{"n"}\Vert\snstar$ & 0 & 15{,}700 & permanent, per withdrawal \\
parameters \texttt{pp} & 4096 & 1{,}641{,}700 & permanent, deploy-time \\
transport $\texttt{"p"}\Vert\snstar$ & $1088 + 3104r$ & $2500 + 400(1121 + 3104r)$ & flow, refunded \\
challenge $\texttt{"w"}\Vert\snstar$ & 1024 & 425{,}300 & flow, refunded \\
state $\texttt{"s"}\Vert\snstar$ & 104 & 57{,}300 & flow, refunded \\
marker $\texttt{"v"}\Vert\snstar$ & 8 & 18{,}900 & flow, app-funded, deleted \\
\bottomrule
\end{tabularx}
\end{table}

\begin{figure}[htbp]
\centering
\begin{tikzpicture}
\begin{axis}[
    width=10.64cm,
    height=6.08cm,
    xlabel={Ring Size $r$},
    ylabel={Size (bytes)},
    xmin=1, xmax=10,
    ymin=0, ymax=34000,
    grid=major,
    legend pos=north west,
]
\addplot[
    color=sciMagenta,
    thick,
    mark=*,
    mark options={solid, fill=sciMagenta},
]
coordinates {
(1,3104) (2,6176) (3,9248) (4,12320) (5,15392)
(6,18464) (7,21536) (8,24608) (9,27680) (10,30752)
};
\addlegendentry{Signature $\sigma$ ($32 + 3072r$)}
\addplot[
    color=sciBlue,
    thick,
    dashed,
    mark=square*,
    mark options={solid, fill=sciBlue},
]
coordinates {
(1,4160) (2,7264) (3,10368) (4,13472) (5,16576)
(6,19680) (7,22784) (8,25888) (9,28992) (10,32096)
};
\addlegendentry{Transport payload ($1056 + 3104r$)}
\end{axis}
\end{tikzpicture}
\caption{Proof and transport sizes versus ring size, computed from the wire format of Section~\ref{sec:crypto}. Every payload exceeds the AVM's 2\,048-byte application-argument limit, necessitating box-streamed transport; the transport box adds the serial-number transform, the ring identifiers, and a funder header to the signature itself.}
\label{fig:proof_transport_size}
\end{figure}
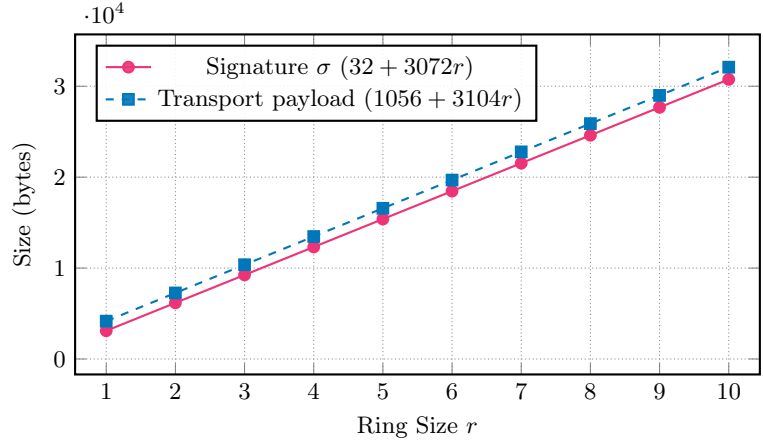

\begin{figure}[htbp]
\centering
\begin{tikzpicture}
\begin{axis}[
    width=10.64cm,
    height=6.08cm,
    xlabel={Ring Size $r$},
    ylabel={Fee (ALGO)},
    xmin=1, xmax=10,
    ymin=0, ymax=20,
    grid=major,
    legend pos=north west,
]
\addplot[
    color=sciTeal,
    thick,
    mark=*,
    mark options={solid, fill=sciTeal},
]
coordinates {
(1,1.828) (2,3.656) (3,5.484) (4,7.312) (5,9.140)
(6,10.968) (7,12.796) (8,14.624) (9,16.452) (10,18.280)
};
\addlegendentry{Inner OpUp fees ($1.828r$)}
\addplot[
    color=sciRed,
    thick,
    dashed,
    mark=square*,
    mark options={solid, fill=sciRed},
]
coordinates {
(1,1.92) (2,3.78) (3,5.64) (4,7.50) (5,9.36)
(6,11.22) (7,13.08) (8,14.94) (9,16.80) (10,18.66)
};
\addlegendentry{Fee envelope ($0.06 + 1.86r$)}
\end{axis}
\end{tikzpicture}
\caption{Withdrawal verification fees versus ring size, computed from the AVM's deterministic fee schedule at the 0.001-ALGO minimum fee. The lower curve is the inner OpUp cost of the $9r$ verification phases ($1828$ calls per member); the upper curve is the envelope used to derive the denomination, additionally covering the outer phase, reference-carrier, transport, and settlement transactions. The actual fronted fee is bound into the signature context $\mu$ and reimbursed from the pool at settlement.}
\label{fig:fees}
\end{figure}
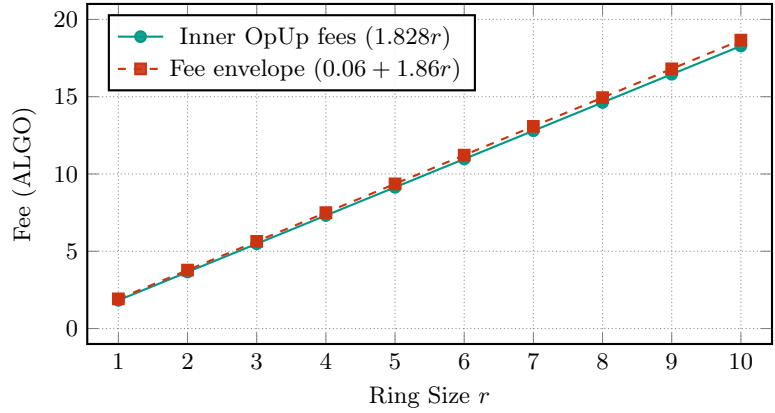

\paragraph{Scalability Trade-off.}
Both cost curves are linear in $r$, so the protocol cannot approach the pool-wide anonymity sets of succinct-proof pools. With the three-response wire format, the cap $r_{\max} = 10$ is a \emph{hard structural} limit rather than merely an economic one: the single transport box holds $1088 + 3104r$ bytes, which is $32\,128$ B at $r = 10$ and thus fits within Algorand's $32\,768$-byte per-box ceiling. The on-chain arithmetic is correct for any $r$, and per-member fees plus the growing transport MBR make rings near the cap costly in any case. Raising it would require splitting the transport across multiple boxes (each independently bounded by 32\,768 bytes), a straightforward but currently unimplemented extension. Off-chain signing is $O(r)$ ring operations plus a small expected number of rejection restarts and completes in milliseconds to seconds. As shown in Figure~\ref{fig:proof_time}, the proof generation time remains sub-second across all measured ring sizes, indicating that the computational burden on the client's local prover is negligible compared to the on-chain verification workload. On-chain verification spans $9r + O(1)$ sequential transaction groups and therefore finishes within a correspondingly bounded number of blocks. The wall-clock latencies in Figure~\ref{fig:runtime_total}, broken down in Figures~\ref{fig:deposit_breakdown} and~\ref{fig:withdrawal_breakdown}, are Algorand testnet measurements dominated by transaction submission and network confirmation; they characterize the deployment environment rather than the protocol itself, and we report them for completeness rather than as an intrinsic protocol benchmark.

\subsection{Latency Breakdown}
To better understand the end-to-end performance characteristics of the protocol, the total wall-clock latencies are decomposed into discrete operational stages measured directly within the client application. 

As illustrated in Figure~\ref{fig:deposit_breakdown}, the deposit workflow consists of four sequential stages:
\begin{enumerate}
    \item \textbf{Deposit Detail Generation:} The client requests a new coin from the local prover enclave, which samples the coin secret $(k,s)$, derives the deterministic blinding $e$ (storing the full tuple $(k,s,e)$), and computes the Ring-SIS commitment $C$.
    \item \textbf{Transaction Construction:} The client hashes the commitment to derive the box storage key and constructs the atomic transaction group, pairing the \textsf{deposit} application call with the fixed-denomination payment.
    \item \textbf{Transaction Submission:} The client delegates the transaction group to the wallet provider for cryptographic signing and submits the signed group to the Algorand node's REST API.
    \item \textbf{Network Confirmation:} The client polls the node, waiting for the submitted transaction group to be committed to a block by the Algorand consensus protocol.
\end{enumerate}
Because the deposit phase involves only a constant-time commitment generation and a single atomic transaction group, all latency components remain independent of the ring size $r$.

Conversely, Figure~\ref{fig:withdrawal_breakdown} details the five stages of the withdrawal workflow:
\begin{enumerate}
    \item \textbf{Withdrawal Data Generation:} The client queries the blockchain indexer for recent commitments, selects a recency-biased decoy pool, inserts the user's true commitment, and applies a Fisher--Yates shuffle to form the anonymity set.
    \item \textbf{Proof Generation:} The client delegates the ring and settlement context to the local prover enclave to produce the lattice ring signature.
    \item \textbf{Transaction Construction:} The client calculates the dynamic fee budget and constructs the sequence of transaction groups required for verification: the \textsf{init\_proof} box allocation, the \textsf{put\_proof} chunked payload streaming, the $9r$ \textsf{prep\_ntt} phases with their inner \texttt{opup} calls and auxiliary reference carriers, and the final \textsf{withdraw} settlement group.
    \item \textbf{Transaction Submission:} The client delegates the entire sequence of transaction groups to the wallet provider for signing in a single prompt, and then sequentially submits them to the Algorand node.
    \item \textbf{Network Confirmation:} The client polls the node, waiting for the final \textsf{withdraw} transaction group to be confirmed, which guarantees that the entire verification pipeline succeeded and the funds were settled.
\end{enumerate}
As shown in Figure~\ref{fig:withdrawal_breakdown}, the off-chain stages (Data Generation, Proof Generation, and Transaction Construction) execute in sub-second or near-second time. The linear growth in total withdrawal latency is driven entirely by the on-chain stages (Transaction Submission and Network Confirmation), which must sequentially process the $9r + O(1)$ transaction groups required to overcome the AVM's opcode limits.

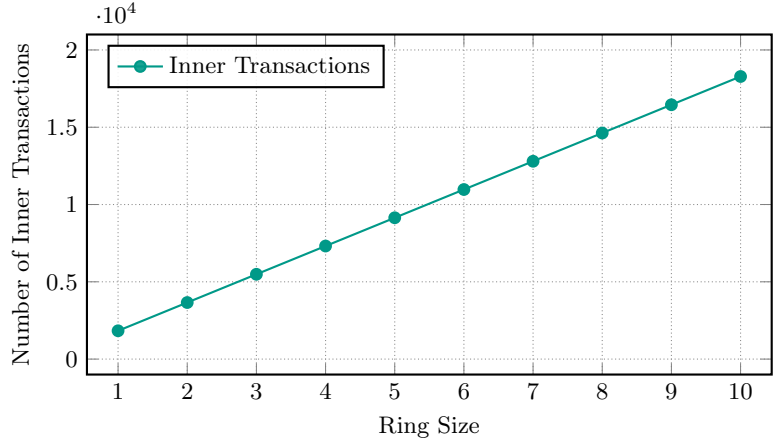
\begin{figure}[htbp]
\centering
\begin{tikzpicture}
\begin{axis}[
    width=10.64cm,
    height=6.08cm,
    xlabel={Ring Size},
    ylabel={Number of Inner Transactions},
    xmin=1, xmax=10,
    ymin=0, ymax=20000,
    xtick={1,2,3,4,5,6,7,8,9,10},
    grid=major,
    legend pos=north west,
]

\addplot[
    color=sciTeal,
    thick,
    mark=*,
    mark options={solid, fill=sciTeal},
]
coordinates {
(1,1831.00) (2,3659.00) (3,5487.00) (4,7315.00) (5,9143.00)
(6,10971.00) (7,12799.00) (8,14627.00) (9,16455.00) (10,18283.00)
};
\addlegendentry{Inner Transactions}

\end{axis}
\end{tikzpicture}
\caption{Number of inner transactions issued during the withdrawal phase versus ring size. To overcome the strict per-transaction AVM opcode budget, the verifier splits the NTT-based lattice ring signature check into nine phases per ring member and dynamically provisions execution budget through grouped inner \texttt{opup} calls. Each ring member requires 1828 inner transactions ($4 \times (250 + 175) + 128$), leading to an $O(r)$ scaling in the total number of inner transactions.}
\label{fig:inner_tx}
\end{figure}

\begin{figure}[htbp]
\centering
\begin{tikzpicture}
\begin{axis}[
    width=10.64cm,
    height=6.08cm,
    xlabel={Ring Size},
    ylabel={Transaction Fee (ALGO)},
    xmin=1, xmax=10,
    ymode=log,
    log basis y=10,
    ymin=0.0005, ymax=50,
    xtick={1,2,3,4,5,6,7,8,9,10},
    grid=major,
    legend pos=north west,
]

\addplot[
    color=sciBlue,
    thick,
    mark=*,
    mark options={solid, fill=sciBlue},
]
coordinates {
(1,0.001) (2,0.001) (3,0.001) (4,0.001) (5,0.001)
(6,0.001) (7,0.001) (8,0.001) (9,0.001) (10,0.001)
};
\addlegendentry{Deposit Fee}

\addplot[
    color=sciRed,
    thick,
    dashed,
    mark=square*,
    mark options={solid, fill=sciRed},
]
coordinates {
(1,1.92) (2,3.78) (3,5.64) (4,7.50) (5,9.36)
(6,11.22) (7,13.08) (8,14.94) (9,16.80) (10,18.66)
};
\addlegendentry{Withdrawal Fee}

\end{axis}
\end{tikzpicture}
\caption{Deposit and withdrawal transaction fees as a function of ring size (logarithmic fee axis). The deposit fee remains constant at 0.001 ALGO, since depositing only requires a standard payment transaction and a box allocation for the coin commitment. Conversely, the withdrawal fee scales linearly with the ring size: the on-chain lattice verification provisions its AVM opcode budget through 1828 inner \texttt{opup} calls per ring member ($\approx 1.828$ ALGO each), plus phase and auxiliary reference-carrier transactions, raising the total fee from 1.92 ALGO at $r=1$ to 18.66 ALGO at $r=10$.}
\label{fig:tx_fee}
\end{figure}
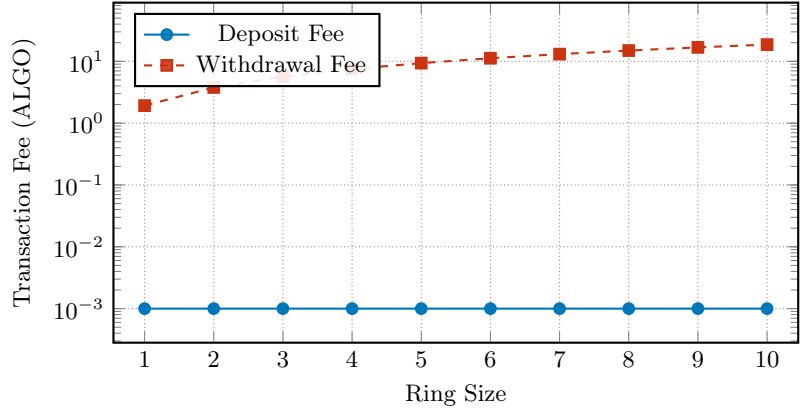

\begin{figure}[htbp]
\centering
\begin{tikzpicture}
\begin{axis}[
    width=10.64cm,
    height=6.08cm,
    xlabel={Ring Size},
    ylabel={Time (ms)},
    xmin=1, xmax=10,
    ymode=log,
    log basis y=10,
    ymin=3000, ymax=600000,
    xtick={1,2,3,4,5,6,7,8,9,10},
    grid=major,
    legend pos=north west,
]

\addplot[
    color=sciOrange, 
    thick, 
    mark=*, 
    mark options={solid, fill=sciOrange},
    error bars/.cd, y dir=both, y explicit
]
coordinates {
(1,5181.80) +- (0,498.69) (2,6100.00) +- (0,925.36) (3,5574.00) +- (0,1067.26) 
(4,5955.20) +- (0,800.33) (5,5580.60) +- (0,654.15) (6,6840.00) +- (0,375.66) 
(7,6222.00) +- (0,1036.36) (8,5744.20) +- (0,846.98) (9,5948.20) +- (0,654.81) 
(10,6042.80) +- (0,1129.45)
};
\addlegendentry{Deposit Total Time}

\addplot[
    color=sciPurple, 
    thick, 
    dashed, 
    mark=square*, 
    mark options={solid, fill=sciPurple},
    error bars/.cd, y dir=both, y explicit
]
coordinates {
(1,57967.00) +- (0,1145.33) (2,98075.60) +- (0,1425.97) (3,144000.80) +- (0,3053.66) 
(4,186444.20) +- (0,4123.52) (5,229632.80) +- (0,5538.93) (6,273531.40) +- (0,5510.06) 
(7,327581.20) +- (0,9070.81) (8,374401.60) +- (0,11666.01) (9,421190.00) +- (0,13396.08) 
(10,468469.00) +- (0,15692.68)
};
\addlegendentry{Withdrawal Total Time}

\end{axis}
\end{tikzpicture}
\caption{End-to-end deposit and withdrawal latency across varying ring sizes (logarithmic time axis). The deposit process exhibits a stable latency profile of roughly 5--7 s regardless of the anonymity set size, as its operations are constant-time. In contrast, withdrawal latency grows linearly from about 58 s at $r=1$ to 468 s at $r=10$, driven not by off-chain signing but by the on-chain verification workload: submitting the $9r$ phase groups and their $1828r$ inner \texttt{opup} transactions required to cover the lattice verification opcode cost.}
\label{fig:runtime_total}
\end{figure}
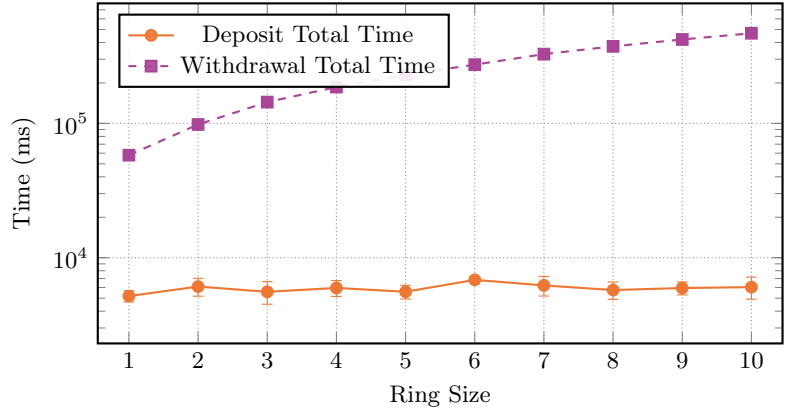

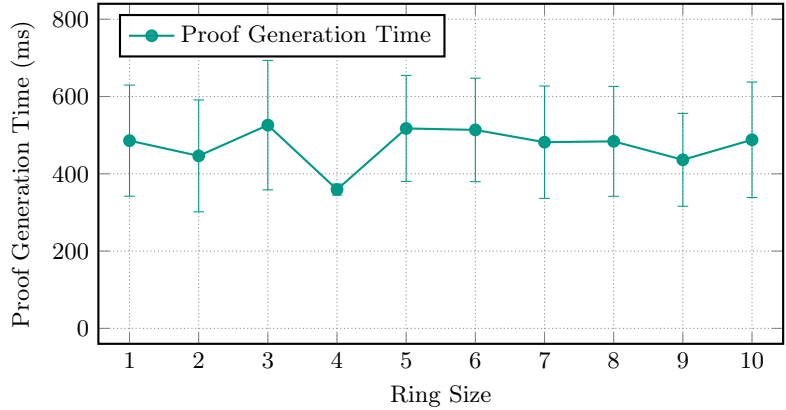
\begin{figure}[htbp]
\centering
\begin{tikzpicture}
\begin{axis}[
    width=10.64cm,
    height=6.08cm,
    xlabel={Ring Size},
    ylabel={Proof Generation Time (ms)},
    xmin=1, xmax=10,
    ymin=0, ymax=800,
    xtick={1,2,3,4,5,6,7,8,9,10},
    grid=major,
    legend pos=north west,
]

\addplot[
    color=sciTeal, 
    thick, 
    mark=*, 
    mark options={solid, fill=sciTeal},
    error bars/.cd, y dir=both, y explicit
]
coordinates {
(1,485.80) +- (0,143.79) (2,446.40) +- (0,144.79) (3,526.00) +- (0,167.54) 
(4,359.20) +- (0,14.34) (5,517.40) +- (0,137.21) (6,513.60) +- (0,133.93) 
(7,481.80) +- (0,145.38) (8,484.00) +- (0,142.23) (9,436.20) +- (0,120.15) 
(10,488.00) +- (0,149.50)
};
\addlegendentry{Proof Generation Time}

\end{axis}
\end{tikzpicture}
\caption{Post-quantum ring signature proof generation time versus ring size. The off-chain local prover executes the AOS/Borromean-style \textsf{Sign} algorithm over the cyclotomic ring $R_q = \mathbb{Z}_q[X]/(X^{512}+1)$, $q = 12289$. Unlike the on-chain verification cost, signing remains sub-second (mean values between 359.20 ms and 526.00 ms) across all measured ring sizes; the run-to-run variation stems mainly from norm-based rejection sampling, which restarts the signing attempt until all responses satisfy the short-vector norm bounds.}
\label{fig:proof_time}
\end{figure}

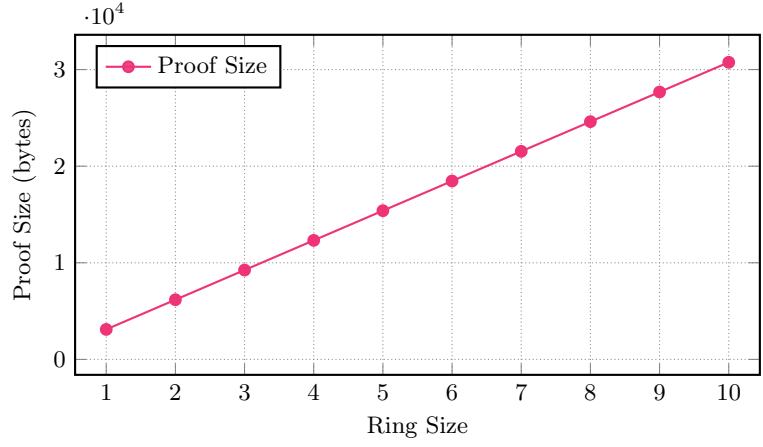
\begin{figure}[htbp]
\centering
\begin{tikzpicture}
\begin{axis}[
    width=10.64cm,
    height=6.08cm,
    xlabel={Ring Size},
    ylabel={Proof Size (bytes)},
    xmin=1, xmax=10,
    ymin=0, ymax=32000,
    xtick={1,2,3,4,5,6,7,8,9,10},
    grid=major,
    legend pos=north west,
]

\addplot[
    color=sciMagenta, 
    thick, 
    mark=*, 
    mark options={solid, fill=sciMagenta}
]
coordinates {
(1,3104.00) (2,6176.00) (3,9248.00) (4,12320.00) (5,15392.00)
(6,18464.00) (7,21536.00) (8,24608.00) (9,27680.00) (10,30752.00)
};
\addlegendentry{Proof Size}

\end{axis}
\end{tikzpicture}
\caption{Proof size growth as a function of ring size. The packed signature $\sigma = (c_0, \{(z_{k,i}, z_{s,i}, z_{e,i})\}_{i=0}^{r-1})$ comprises a 32-byte initial challenge and, per ring member, three 512-coefficient response polynomials stored at 2 bytes per coefficient, yielding a deterministic size of $3072r + 32$ bytes. Since this payload far exceeds application-argument limits, it is uploaded in chunks into a dedicated proof box, and the smart contract shape-checks the ring size to $r \in [1, 10]$, dictating the practical upper bound of the anonymity set.}
\label{fig:proof_size}
\end{figure}

\begin{figure}[htbp]
\centering
\begin{tikzpicture}
\begin{axis}[
    width=10.64cm,
    height=6.08cm,
    xlabel={Ring Size},
    ylabel={Time (ms)},
    xmin=1, xmax=10,
    ymin=0, ymax=8000,
    xtick={1,2,3,4,5,6,7,8,9,10},
    grid=major,
    legend style={at={(0.5,-0.22)}, anchor=north, legend columns=2},
]

\addplot[
    color=sciBlue, thick, mark=*, mark options={solid, fill=sciBlue},
    error bars/.cd, y dir=both, y explicit
] coordinates {
(1,143.40) +- (0,167.61) (2,328.00) +- (0,6.04) (3,83.60) +- (0,136.63) 
(4,268.00) +- (0,137.09) (5,139.60) +- (0,163.40) (6,144.20) +- (0,169.20) 
(7,326.20) +- (0,4.21) (8,143.00) +- (0,168.45) (9,205.20) +- (0,167.34) 
(10,205.40) +- (0,169.33)
};
\addlegendentry{Deposit Detail Generation}

\addplot[
    color=sciCyan, thick, densely dotted, mark=square*, mark options={solid, fill=sciCyan},
    error bars/.cd, y dir=both, y explicit
] coordinates {
(1,86.60) +- (0,7.99) (2,43.80) +- (0,25.29) (3,47.80) +- (0,29.94) 
(4,40.60) +- (0,13.20) (5,54.40) +- (0,27.37) (6,46.80) +- (0,29.83) 
(7,54.20) +- (0,36.13) (8,41.40) +- (0,17.21) (9,46.40) +- (0,27.20) 
(10,44.60) +- (0,24.32)
};
\addlegendentry{Transaction Construction}

\addplot[
    color=sciOrange, thick, dashed, mark=triangle*, mark options={solid, fill=sciOrange},
    error bars/.cd, y dir=both, y explicit
] coordinates {
(1,1505.80) +- (0,103.78) (2,1538.60) +- (0,102.11) (3,1548.00) +- (0,195.39) 
(4,1432.20) +- (0,110.66) (5,1448.80) +- (0,115.75) (6,1579.60) +- (0,187.31) 
(7,1553.60) +- (0,100.11) (8,1669.60) +- (0,59.85) (9,1699.00) +- (0,151.20) 
(10,1654.80) +- (0,199.41)
};
\addlegendentry{Transaction Submission}

\addplot[
    color=sciPurple, thick, dashdotted, mark=diamond*, mark options={solid, fill=sciPurple},
    error bars/.cd, y dir=both, y explicit
] coordinates {
(1,3445.40) +- (0,368.98) (2,4189.40) +- (0,883.72) (3,3894.00) +- (0,990.74) 
(4,4214.00) +- (0,932.07) (5,3937.00) +- (0,597.05) (6,5068.60) +- (0,454.37) 
(7,4287.40) +- (0,1059.78) (8,3889.80) +- (0,741.96) (9,3997.00) +- (0,699.10) 
(10,4137.40) +- (0,947.51)
};
\addlegendentry{Network Confirmation}

\addplot[
    color=sciRed, ultra thick, mark=x, mark options={solid},
    error bars/.cd, y dir=both, y explicit
] coordinates {
(1,5181.80) +- (0,498.69) (2,6100.00) +- (0,925.36) (3,5574.00) +- (0,1067.26) 
(4,5955.20) +- (0,800.33) (5,5580.60) +- (0,654.15) (6,6840.00) +- (0,375.66) 
(7,6222.00) +- (0,1036.36) (8,5744.20) +- (0,846.98) (9,5948.20) +- (0,654.81) 
(10,6042.80) +- (0,1129.45)
};
\addlegendentry{Total Deposit Time}

\end{axis}
\end{tikzpicture}
\caption{Breakdown of deposit latency versus ring size. The total deposit time is decomposed into four discrete operational stages: deposit detail generation, transaction construction, transaction submission, and network confirmation. Because the deposit phase involves only a constant-time Ring-SIS commitment generation and a standard atomic transaction group with a single box allocation, all latency components remain independent of the ring size $r$, with the total dominated by network confirmation and transaction submission.}
\label{fig:deposit_breakdown}
\end{figure}
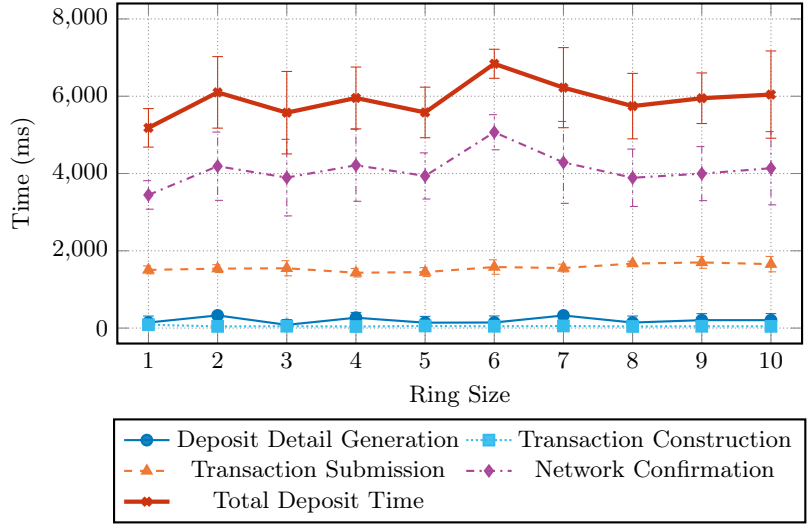

\begin{figure}[t]
\centering
\begin{tikzpicture}
\begin{axis}[
    width=10.64cm,
    height=6.08cm,
    xlabel={Ring Size},
    ylabel={Time (ms)},
    xmin=1, xmax=10,
    ymode=log,
    log basis y=10,
    ymin=10, ymax=1000000,
    xtick={1,2,3,4,5,6,7,8,9,10},
    grid=major,
    legend style={at={(0.5,-0.22)}, anchor=north, legend columns=2},
]

\addplot[
    color=sciTeal, thick, mark=*, mark options={solid, fill=sciTeal},
    error bars/.cd, y dir=both, y explicit
] coordinates {
(1,511.40) +- (0,188.64) (2,716.20) +- (0,159.24) (3,769.40) +- (0,137.73) 
(4,1125.00) +- (0,493.73) (5,993.80) +- (0,502.05) (6,795.60) +- (0,174.11) 
(7,887.80) +- (0,178.17) (8,1030.20) +- (0,194.90) (9,1008.20) +- (0,138.74) 
(10,1186.60) +- (0,524.01)
};
\addlegendentry{Withdrawal Data Generation}

\addplot[
    color=sciBlue, thick, densely dotted, mark=square*, mark options={solid, fill=sciBlue},
    error bars/.cd, y dir=both, y explicit
] coordinates {
(1,485.80) +- (0,143.79) (2,446.40) +- (0,144.79) (3,526.00) +- (0,167.54) 
(4,359.20) +- (0,14.34) (5,517.40) +- (0,137.21) (6,513.60) +- (0,133.93) 
(7,481.80) +- (0,145.38) (8,484.00) +- (0,142.23) (9,436.20) +- (0,120.15) 
(10,488.00) +- (0,149.50)
};
\addlegendentry{Proof Generation}

\addplot[
    color=sciMagenta, thick, dashed, mark=triangle*, mark options={solid, fill=sciMagenta},
    error bars/.cd, y dir=both, y explicit
] coordinates {
(1,38.40) +- (0,4.22) (2,39.20) +- (0,5.97) (3,44.40) +- (0,12.52) 
(4,56.00) +- (0,16.78) (5,47.60) +- (0,6.58) (6,50.40) +- (0,20.01) 
(7,60.60) +- (0,35.04) (8,199.60) +- (0,319.55) (9,92.80) +- (0,31.35) 
(10,67.40) +- (0,27.01)
};
\addlegendentry{Transaction Construction}

\addplot[
    color=sciOrange, thick, dashdotted, mark=diamond*, mark options={solid, fill=sciOrange},
    error bars/.cd, y dir=both, y explicit
] coordinates {
(1,51535.20) +- (0,1116.61) (2,91476.80) +- (0,1432.24) (3,137268.20) +- (0,3011.71) 
(4,179521.00) +- (0,3978.87) (5,222678.60) +- (0,5555.00) (6,266790.80) +- (0,5525.00) 
(7,320762.00) +- (0,9081.59) (8,367311.60) +- (0,11724.33) (9,414298.60) +- (0,13381.43) 
(10,461330.40) +- (0,15627.16)
};
\addlegendentry{Transaction Submission}

\addplot[
    color=sciPurple, thick, mark=pentagon*, mark options={solid, fill=sciPurple},
    error bars/.cd, y dir=both, y explicit
] coordinates {
(1,5366.20) +- (0,16.38) (2,5362.80) +- (0,44.70) (3,5359.00) +- (0,19.27) 
(4,5350.40) +- (0,23.48) (5,5364.00) +- (0,14.76) (6,5349.60) +- (0,10.43) 
(7,5359.40) +- (0,10.90) (8,5347.20) +- (0,39.00) (9,5324.20) +- (0,23.15) 
(10,5362.80) +- (0,5.07)
};
\addlegendentry{Network Confirmation}

\addplot[
    color=sciGray, ultra thick, dashed, mark=x, mark options={solid},
    error bars/.cd, y dir=both, y explicit
] coordinates {
(1,57967.00) +- (0,1145.33) (2,98075.60) +- (0,1425.97) (3,144000.80) +- (0,3053.66) 
(4,186444.20) +- (0,4123.52) (5,229632.80) +- (0,5538.93) (6,273531.40) +- (0,5510.06) 
(7,327581.20) +- (0,9070.81) (8,374401.60) +- (0,11666.01) (9,421190.00) +- (0,13396.08) 
(10,468469.00) +- (0,15692.68)
};
\addlegendentry{Total Withdrawal Time}

\end{axis}
\end{tikzpicture}
\caption{Breakdown of withdrawal latency versus ring size (logarithmic time axis). The off-chain stages (withdrawal data generation, proof generation, and transaction construction) all remain at or below roughly one second irrespective of the ring size. The end-to-end time is instead dominated by transaction submission, which grows linearly from about 52 s at $r=1$ to 461 s at $r=10$ as the client submits the $9r$ on-chain verification phase groups carrying $1828r$ inner \texttt{opup} transactions, while the final network confirmation wait stays constant at roughly 5.3 s.}
\label{fig:withdrawal_breakdown}
\end{figure}
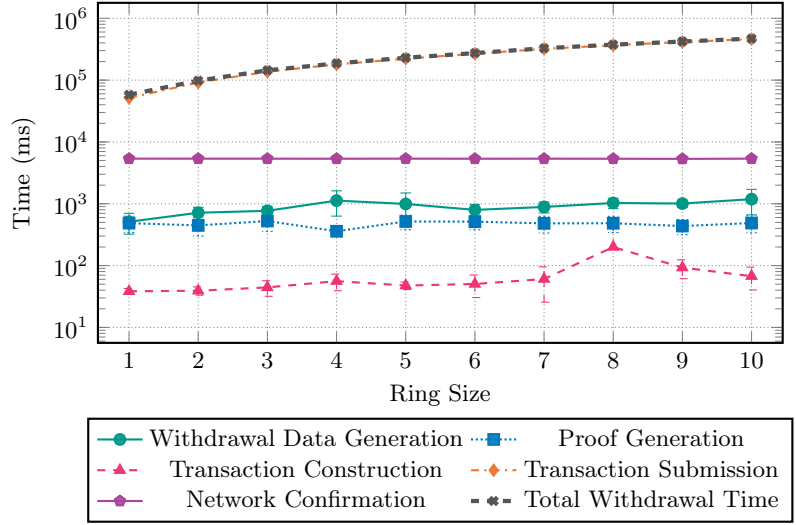

\section{Conclusion}
\label{sec:concl}
We presented Obscura-PQ, a decentralized, non-custodial privacy protocol that demonstrates the practical feasibility of post-quantum cryptography on highly constrained public blockchains. By evaluating all relations in the number-theoretic-transform (NTT) domain, splitting table-driven NTTs across opcode-pooled execution phases, and leveraging box-streamed proof transport, we successfully deployed a setup-free lattice linkable ring signature natively on the Algorand testnet. Our construction binds a Ring-SIS commitment to a deterministic Ring-LWE serial number using an AOS/Borromean-style challenge chain, achieving double-spend soundness, unforgeability, and anonymity in the classical random-oracle model without relying on a global Merkle accumulator.

Several important directions remain for future work. On the cryptographic front, a formal security analysis in the quantum-accessible random-oracle model (QROM) \cite{don2019security} and the calibration of concrete bit-security parameters using formal lattice estimators \cite{albrecht2015concrete} are necessary next steps. Furthermore, a production deployment requires a fully hardened, constant-time discrete Gaussian sampler with the complete Fiat--Shamir-with-aborts acceptance ratio \cite{lyubashevsky2012lattice}. On the protocol side, exploring sublinear post-quantum accumulators or logarithmic ring signatures \cite{beullens2020calamari,ben2018scalable} could enable significantly larger anonymity sets, provided their verifiers remain compatible with our phase-decomposed execution model. Additionally, integrating confidential amounts in the style of MatRiCT \cite{esgin2019matrict} and establishing a third-party relayer network utilizing our signature-bound relayer role would greatly expand the protocol's utility. Finally, formal verification of the TEAL phase machine and independent security audits remain critical prerequisites for mainnet deployment.

\section*{Code Availability}
To facilitate reproducibility and public auditing, the complete source code is available at \url{https://github.com/n-azimi/Obscura-PQ}. The repository includes the PyTeal smart contract, local lattice prover, React client, test suites, and ledger-analysis utilities. All empirical measurements reported in Section~\ref{sec:eval} were derived directly from Algorand testnet executions of this implementation.

\appendix
\renewcommand{\theHsection}{appendix.\Alph{section}}
\section{Formal Definitions and Security Arguments}
\label{app:formal}
This appendix formalizes the lattice linkable ring signature underlying Obscura-PQ and relates each security property to the assumptions of Section~\ref{sec:prelim}. The arguments follow established analysis templates for Fiat--Shamir-with-aborts sigma protocols and commitment-based linkable ring signatures \cite{lyubashevsky2012lattice,esgin2019lattice,alberto2018post,yuen2021dualring}. We present the reductions, indicating precisely where each assumption enters, and we flag the points at which the implemented scheme deviates from the analyzed abstraction (Remark~\ref{rem:gaps}). All reductions are in the classical random-oracle model. We are deliberately conservative about what is proven: the extractor recovers a \emph{relaxed} opening rather than an exact one (Theorem~\ref{thm:unforgeability}). The anonymity argument reduces to decisional Ring-LWE \emph{in the random-oracle model}: the serial number $\sn = a_3 k + e$ has the form of a Ring-LWE sample, and modelling $H_{\mathrm{short}}$ as a random oracle makes its \emph{deterministic} error $e = H_{\mathrm{short}}(k)$ behave as an independent short error unless the adversary queries the oracle at the secret $k$ (an event bounded by key recovery). We flag this self-correlated-noise caveat explicitly in Remark~\ref{rem:gaps}.

\subsection{Notation}
Table~\ref{tab:notation} fixes notation. Throughout, ``short'' means bounded Euclidean norm, and all hash calls are domain-separated instances of SHA-256 modeled as random oracles.

\begin{table}[htbp]
\centering
\caption{Summary of notation.}
\label{tab:notation}
\begin{tabularx}{\textwidth}{@{}l X@{}}
\toprule
\textbf{Symbol} & \textbf{Description} \\
\midrule
$\Rq$, $n$, $q$ & quotient ring $\mathbb{Z}_q[X]/(X^n+1)$, $n = 512$, $q = 12289$ \\
$a_1, a_2, a_3, a_4$ & nothing-up-my-sleeve public parameters in $\Rq$ \\
$(k, s, e)$, $C = a_1 k + a_2 s + a_4 e$ & coin opening (short; $e = H_{\mathrm{short}}(k)$) and rank-one Ring-SIS coin commitment \\
$\sn = a_3 k + e$, $\snstar$ & serial number (key image) and its SHA-256 nullifier digest \\
$\mathbf{C}$, $r$, $\pi$ & ring of commitments, ring size, secret signer index \\
$\mu$ & context: recipient $\Vert$ relayer $\Vert$ fee $\Vert$ appID (80 bytes) \\
$\mathcal{C}$, $w$ & sparse ternary challenge space, Hamming weight $48$ \\
$c_i$, $\chi_i = \HtoC(c_i)$ & 256-bit chain values and derived ternary challenges \\
$y_k, y_s, y_e$; $z_{k,i}, z_{s,i}, z_{e,i}$ & Gaussian masks (width $\varsigma = 2500$); per-member responses \\
$\beta^2$ & squared-norm acceptance bound, $4.2 \times 10^{9}$ \\
$\sigma = (c_0, \{(z_{k,i}, z_{s,i}, z_{e,i})\}_i)$ & signature \\
$\mathcal{B}_C$, $\mathcal{B}_N$ & commitment and nullifier box sets \\
$\hat{x} = \NTT(x)$, $\odot$ & frequency-domain image and pointwise product \\
\bottomrule
\end{tabularx}
\end{table}

\subsection{Hardness Assumptions and the Commitment}
\label{app:assumptions}
We restate the assumptions of Section~\ref{sec:prelim} in the precise forms the reductions use. Throughout, the public parameters $a_1, a_2, a_3, a_4$ are treated as uniform in $\Rq$; this is the nothing-up-my-sleeve heuristic, justified in the ROM by modeling $H_{\mathrm{ring}}$ as a (programmable) random oracle, so a reduction may embed a challenge instance in the $a_j$ by reprogramming $H_{\mathrm{ring}}$. ``Short'' means squared Euclidean norm at most a stated bound.

\begin{assumption}[Ring-SIS, binding; B1]
\label{ass:sis}
For a vector $\mathbf{a}$ of uniform elements in $\Rq$ of length 3, and a norm bound $B = O(\beta)$, no PPT algorithm outputs a nonzero short vector $\mathbf{f}$ such that $\langle \mathbf{a}, \mathbf{f}\rangle = 0$ except with negligible probability. In particular, the coin $C = a_1 k + a_2 s + a_4 e$ is \emph{binding}: a second short opening $(k',s',e') \neq (k,s,e)$ yields the collision $a_1(k-k') + a_2(s-s') + a_4(e-e') = 0$, which is a Ring-SIS solution for $\mathbf{a} = (a_1, a_2, a_4)$.
\end{assumption}

\begin{assumption}[Ring-LWE, commitment and serial-number hiding; B2]
\label{ass:lwe}
For uniform $a$ and short $x, y$, the sample $(a,\, a x + y)$ is computationally indistinguishable from uniform on $\Rq \times \Rq$ (decisional Ring-LWE \cite{lyubashevsky2010ideal}). Two consequences are used. (i) The coin $C = a_1 k + a_2 s + a_4 e$ can be rewritten as $a_2\bigl((a_2^{-1}a_1)k + (a_2^{-1}a_4)e + s\bigr)$ whenever $a_2$ is invertible. Since $a_1, a_4$ are uniform, $a_2^{-1}a_1$ and $a_2^{-1}a_4$ are uniform, making the inner term a rank-2 Module-LWE sample with secret $(k,e)^T$ and error $s$. Thus, $C$ is computationally indistinguishable from uniform and hides its opening. (ii) The serial number $\sn = a_3 k + e$ has the form of a Ring-LWE sample, but with the \emph{deterministic} error $e = H_{\mathrm{short}}(k)$. Modelling $H_{\mathrm{short}}$ as a random oracle with short outputs, for any adversary that does not query the oracle at $k$ the value $e$ is a uniform short element independent of $k$; then $\sn$ is distributed exactly as a genuine Ring-LWE sample and, jointly with $C$ (which reuses the same $e$), the pair is distributed as in the independent-error variant. Under this modelling $\sn$ hides $k$, recovering $k$ from $\sn$ is search Ring-LWE (equivalently bounded-distance decoding of $a_3 k + e = \sn$ with short $(k,e)$), and the residual event that the adversary queries $H_{\mathrm{short}}(k)$ is itself a key recovery, bounded by the same search hardness (Remark~\ref{rem:gaps}).
\end{assumption}

\begin{assumption}[Serial-number linking; decisional Ring-LWE in the ROM]
\label{ass:link}
For uniform $a_1, a_2, a_3, a_4$ and independent short $k, k', s, e, e'$, the distributions
\[
    \bigl(a_1 k + a_2 s + a_4 e,\; a_3 k + e\bigr)
    \quad\text{and}\quad
    \bigl(a_1 k + a_2 s + a_4 e,\; a_3 k' + e'\bigr)
\]
are computationally indistinguishable. Equivalently, given a commitment and a candidate serial number, no PPT algorithm decides whether they share the committed secret $k$, except with negligible advantage. For \emph{independent} $e, e'$ this is a direct consequence of Assumption~\ref{ass:lwe}; the deployed scheme instead uses the \emph{deterministic} $e = H_{\mathrm{short}}(k)$, for which the statement holds in the random-oracle model, as discussed below and in Remark~\ref{rem:gaps}.
\end{assumption}

\noindent The blinded serial number $\sn = a_3 k + e$ carries a short error, so the linking statement~\ref{ass:link} holds \emph{in the random-oracle model}: modelling $H_{\mathrm{short}}$ as a random oracle, the joint view $(a_1 k + a_2 s + a_4 e,\; a_3 k + e)$ is, up to the negligible probability that the adversary queries $H_{\mathrm{short}}(k)$, distributed exactly as in the \emph{independent-error} variant, for which replacing $\sn$ by a uniform element (one decisional Ring-LWE hop) and reintroducing an independent $a_3 k' + e'$ (a second hop) makes the two distributions identical. Because $e = H_{\mathrm{short}}(k)$ is \emph{not} an independent error, this is a random-oracle argument, \emph{not} a black-box reduction to standard decisional Ring-LWE; without the ROM, Assumption~\ref{ass:link} should be read as a self-correlated-noise Ring-LWE assumption in its own right (Remark~\ref{rem:gaps}), and the clean standard-model reduction is recovered only by the alternative that samples $e$ independently and stores it in the coin. This is the lattice analogue of the key-image indistinguishability relied on in the commitment-and-serial-number line \cite{alberto2018post,esgin2019matrict,yuen2021dualring}; Section~\ref{sec:prelim} folds both uses of Ring-LWE into the single assumption B2. Recovering the short $k$ from $\sn$ is the search Ring-LWE problem and is likewise covered by Assumption~\ref{ass:lwe}.

\subsection{Scheme Syntax and Correctness}

\begin{definition}[Lattice Linkable Ring Signature]
\label{def:llrs}
Relative to public parameters $\mathrm{pp} = (a_1,a_2,a_3,a_4)$ and a random oracle $\mathcal{H}$, the scheme is the tuple of PPT algorithms $(\CoinKeyGen, \SerNum, \Sig, \Ver, \Link)$:
\begin{itemize}
    \item $\bigl((k,s,e),\, C\bigr) \leftarrow \CoinKeyGen(1^\lambda)$: sample short $k, s$ from the centred binomial distribution, derive the short blinding $e = H_{\mathrm{short}}(k)$, and output the coin $C = a_1 k + a_2 s + a_4 e$;
    \item $\sn \leftarrow \SerNum(k)$: output the deterministic serial number $\sn = a_3 k + e$ with $e = H_{\mathrm{short}}(k)$;
    \item $(\sigma, \sn) \leftarrow \Sig\bigl((k,s,e), \mathbf{C}, \pi, \mu\bigr)$: on input an opening $(k,s,e)$ of $C_\pi \in \mathbf{C}$ and a message $\mu$, output a signature and the serial number, per Section~\ref{sec:sign};
    \item $\{0,1\} \leftarrow \Ver(\mathbf{C}, \mu, \sn, \sigma)$: verify the norm bounds, recompute the challenge chain, and accept iff it closes, per Section~\ref{sec:verify};
    \item $\Link$: two verifying tuples are linked iff their serial-number digests $\snstar$ coincide.
\end{itemize}
\end{definition}

\begin{theorem}[Correctness]
\label{thm:correctness}
For every ring $\mathbf{C}$, index $\pi$, and message $\mu$, an honestly generated signature is accepted by $\Ver$ except with negligible probability over the signer's randomness, and honest signing terminates after an expected $O(1)$ rejection restarts.
\end{theorem}

\begin{proof}
At the true index, substituting the closure $z_{k,\pi} = y_k + \chi_\pi k$, $z_{s,\pi} = y_s + \chi_\pi s$, $z_{e,\pi} = y_e + \chi_\pi e$ into~\eqref{eq:relations} gives
\begin{align*}
t^{\mathrm{pk}}_\pi &= a_1(y_k + \chi_\pi k) + a_2(y_s + \chi_\pi s) + a_4(y_e + \chi_\pi e) - \chi_\pi C_\pi = a_1 y_k + a_2 y_s + a_4 y_e, \\
t^{\mathrm{sn}}_\pi &= a_3(y_k + \chi_\pi k) + (y_e + \chi_\pi e) - \chi_\pi\,(a_3 k + e) = a_3 y_k + y_e ,
\end{align*}
which matches the prover's commitment step; at every decoy index the verifier recomputes the same relation images from the same responses and challenges. Every chain value therefore coincides with the prover's and the chain closes to $c_0$, \emph{provided} every response passes the norm check. The true-index responses pass by construction (rejection sampling). The decoy responses are drawn from $D_\varsigma$ and, in the implemented signer, are \emph{not} conditioned on the bound; a single decoy coefficient vector violates $\lVert \cdot \rVert^2 \le \beta^2$ only with probability negligible in $n$ (see the termination estimate below), so a union bound over the $3(r-1) \le 3 r_{\max}$ decoy responses leaves acceptance overwhelming. This is why the statement is ``except with negligible probability'' rather than perfect; a signer that also norm-checks its decoys would make correctness perfect.

For termination: each coefficient of $\chi_\pi k$ and $\chi_\pi s$ is a signed sum of at most $w$ terms of magnitude at most $\eta$, so $\lVert \chi_\pi k \rVert_\infty \le w\eta = 96 \ll \varsigma$ and $\lVert z_{\bullet,\pi}\rVert^2$ concentrates around the mask mean $n\varsigma^2 = 3.2 \times 10^{9}$. The acceptance bound $\beta^2 = 4.2 \times 10^{9}$ exceeds this mean by roughly five standard deviations of $\lVert z \rVert^2$ (whose standard deviation is $\varsigma^2\sqrt{2n} \approx 2 \times 10^{8}$). The discrete Gaussian rejection sampling completes in expected $M$ attempts, where $M \approx 3$ is chosen based on the worst-case norm of the secret offset, yielding an expected $O(1)$ restarts \cite{lyubashevsky2012lattice}. (The deployed implementation simplifies this by omitting the probability ratio; see Section~\ref{sec:limitations} and Remark~\ref{rem:gaps}.) \qed
\end{proof}

\begin{lemma}[Challenge Space]
\label{lem:challenge}
$\lvert \mathcal{C} \rvert = \binom{512}{48} \cdot 2^{48} > 2^{256}$, and in the ROM each distinct chain input maps to an independent uniform 256-bit value $c_i$; hence a collision among any polynomial number $Q$ of oracle queries, or the prediction of an unqueried chain value, occurs with probability at most $O(Q^2 2^{-256})$ and $Q\,2^{-256}$ respectively, both negligible.
\end{lemma}

\begin{proof}
The count is direct: a challenge is determined by the choice of $w = 48$ positions out of $n = 512$ and a sign per position, and $\log_2\binom{512}{48} + 48 > 256$. Chain values are outputs of distinct random-oracle queries, since the 2-byte index and the packed relation images $\pack(\widehat{t^{\mathrm{pk}}_i}) \Vert \pack(\widehat{t^{\mathrm{sn}}_i})$ separate the inputs; the collision and prediction bounds are then the standard birthday and guessing bounds. The seed $c_i$ carries at most $256$ bits of entropy, so $\HtoC(c_i)$ realizes at most $2^{256}$ of the $\lvert\mathcal{C}\rvert$ challenges, which is the relevant unpredictability parameter. \qed
\end{proof}

\subsection{Security Experiments}
\label{app:games}
We fix the experiments the theorems refer to. In each, the challenger generates $\mathrm{pp}$, answers random-oracle queries by lazy sampling, and maintains a list $L$ of honest coins created through the oracle
$\mathcal{O}^{\mathrm{coin}}$ (runs $\CoinKeyGen$, stores $(k,s,e)$, returns $C$),
a corruption oracle $\mathcal{O}^{\mathrm{corrupt}}(j)$ (returns the opening of the $j$-th honest coin and marks it corrupted),
and a signing oracle $\mathcal{O}^{\mathrm{sign}}(j, \mathbf{C}, \mu)$ (aborts unless $C_j \in \mathbf{C}$; otherwise returns $\Sig$ with the $j$-th opening at its index in $\mathbf{C}$).

\begin{definition}[Unforgeability / Theft Resistance]
\label{def:uf}
$\mathcal{A}$ wins $\mathsf{Exp}^{\mathrm{uf}}$ if, given $\mathrm{pp}$ and the three oracles, it outputs a verifying tuple $(\mathbf{C}^\ast, \mu^\ast, \sn^\ast, \sigma^\ast)$ such that (i) every $C_i \in \mathbf{C}^\ast$ is an honest, uncorrupted coin in $L$, and (ii) $\sn^\ast$ was never returned by $\mathcal{O}^{\mathrm{sign}}$. The scheme is unforgeable if $\Pr[\mathcal{A}\text{ wins}]$ is negligible for every PPT $\mathcal{A}$.
\end{definition}

\begin{definition}[Anonymity / Signer Ambiguity]
\label{def:anon}
For $b \in \{0,1\}$, $\mathsf{Exp}^{\mathrm{anon}}_b$ runs $\mathcal{A}$ with $\mathrm{pp}$ and the three oracles; $\mathcal{A}$ outputs $(\mathbf{C}, \mu, i_0, i_1)$ where $C_{i_0}, C_{i_1} \in \mathbf{C}$ are honest and \emph{never queried} to $\mathcal{O}^{\mathrm{corrupt}}$ or $\mathcal{O}^{\mathrm{sign}}$. The challenger returns $(\sigma, \sn) \leftarrow \Sig$ with the opening of $C_{i_b}$, and $\mathcal{A}$ outputs a guess $b'$. The advantage is $\bigl\lvert \Pr[b'=b] - \tfrac12 \bigr\rvert$. The non-corruption/non-signing restriction on $i_0,i_1$ is intrinsic to linkable ring signatures: two spends of one coin publish the identical serial number and are, by design, linkable.
\end{definition}

\subsection{Unforgeability}
\label{app:uf}

\begin{theorem}[Unforgeability / Theft Resistance]
\label{thm:unforgeability}
In the ROM, under Assumption~\ref{ass:sis} (with $H_{\mathrm{ring}}$ programmable) and Assumption~\ref{ass:lwe}, no PPT adversary wins $\mathsf{Exp}^{\mathrm{uf}}$ except with negligible probability.
\end{theorem}

\begin{proof}
We construct a reduction $\mathcal{B}$ that uses a successful adversary $\mathcal{A}$ to solve the Ring-SIS problem (Assumption~\ref{ass:sis}). Let $Q_H$ be the maximum number of random oracle queries made by $\mathcal{A}$. The reduction $\mathcal{B}$ is given a Ring-SIS challenge $\mathbf{a} = (a_1, a_2, a_4)$ and programs $H_{\mathrm{ring}}$ such that the corresponding public parameters match $\mathbf{a}$, while generating $a_3$ honestly. $\mathcal{B}$ handles $\mathcal{O}^{\mathrm{coin}}$ by generating honest coin openings $(k_i, s_i, e_i)$ and storing them, answering $\mathcal{O}^{\mathrm{corrupt}}$ and $\mathcal{O}^{\mathrm{sign}}$ using these known openings.

Suppose $\mathcal{A}$ outputs a successful forgery $(\mathbf{C}^\ast, \mu^\ast, \sn^\ast, \sigma^\ast)$ with probability $\varepsilon$. The signature $\sigma^\ast$ contains responses and a challenge chain that closes. By Lemma~\ref{lem:challenge}, except with probability at most $Q_H 2^{-256}$, the chain values were obtained via queries to $\mathcal{H}$. By the General Forking Lemma \cite{pointcheval2000security,bellare2006multi}, $\mathcal{B}$ can rewind $\mathcal{A}$ to the oracle query corresponding to the target index $\pi$ (the index where the response relations were first queried), feed it a different random oracle response, and obtain a second accepting transcript with probability $\varepsilon' \ge (\varepsilon - Q_H 2^{-256})^2 / Q_H$.

The two accepting transcripts share the relation images $(t^{\mathrm{pk}}_\pi, t^{\mathrm{sn}}_\pi)$ but have distinct challenges $\chi_\pi \neq \chi'_\pi$. Subtracting the two instances of~\eqref{eq:relations} at $\pi$ yields:
\[
a_1 \bar{z}_k + a_2 \bar{z}_s + a_4 \bar{z}_e = \bar{\chi}\, C_\pi,
\qquad
a_3 \bar{z}_k + \bar{z}_e = \bar{\chi}\, \sn^\ast ,
\]
where $\bar{z}_k = z_{k,\pi} - z'_{k,\pi}$, $\bar{z}_s = z_{s,\pi} - z'_{s,\pi}$, $\bar{z}_e = z_{e,\pi} - z'_{e,\pi}$, and $\bar{\chi} = \chi_\pi - \chi'_\pi$. Since $\chi_\pi, \chi'_\pi \in \mathcal{C}$ are distinct, $\bar{\chi} \neq 0$. Because both transcripts satisfy the norm bounds, $\lVert \bar{z}_k \rVert, \lVert \bar{z}_s \rVert, \lVert \bar{z}_e \rVert \le 2\beta$. This gives a \emph{relaxed} opening of $C_\pi$, which is the standard extracted witness format for lattice-based sigma protocols over rings where short challenge differences may lack inverses \cite{lyubashevsky2012lattice,esgin2019lattice,yuen2021dualring}.

By the rules of $\mathsf{Exp}^{\mathrm{uf}}$, $C_\pi$ is an uncorrupted honest coin generated by $\mathcal{B}$, so $\mathcal{B}$ knows its exact opening $(k_\pi, s_\pi, e_\pi)$ where $C_\pi = a_1 k_\pi + a_2 s_\pi + a_4 e_\pi$ and its true serial number $\sn_\pi = a_3 k_\pi + e_\pi$. Subtracting the known honest relation from the extracted one gives, writing $\delta_k = \bar{z}_k - \bar{\chi} k_\pi$, $\delta_s = \bar{z}_s - \bar{\chi} s_\pi$, $\delta_e = \bar{z}_e - \bar{\chi} e_\pi$:
\[
a_1\delta_k + a_2\delta_s + a_4\delta_e = 0,
\qquad
a_3\delta_k + \delta_e = \bar{\chi}(\sn^\ast - \sn_\pi) .
\]
We now consider three cases:
\begin{enumerate}
    \item $(\delta_k, \delta_s, \delta_e) \neq (0,0,0)$. The first equation directly yields a short, nonzero solution to the Ring-SIS instance for $(a_1, a_2, a_4)$, breaking Assumption~\ref{ass:sis} \cite{ajtai1996generating,lyubashevsky2006generalized}.
    \item $(\delta_k, \delta_s, \delta_e) = (0,0,0)$ and $\sn^\ast \neq \sn_\pi$. Then the second equation implies $\bar{\chi}(\sn^\ast - \sn_\pi) = 0$. In the Fiat--Shamir transform, the serial number $\sn^\ast$ is committed to before the challenge $\chi_\pi$ is generated. For any fixed $\Delta\sn = \sn^\ast - \sn_\pi \neq 0$, $\Delta\sn$ is non-zero modulo at least one prime ideal $\mathfrak{p}$ of $\Rq$. Because the challenge space $\mathcal{C}$ is sufficiently large and well-distributed, the probability over the independent choices of $\chi_\pi, \chi'_\pi \in \mathcal{C}$ that $\chi_\pi \equiv \chi'_\pi \pmod{\mathfrak{p}}$ is strictly bounded away from 1 (and is small for our parameters). Thus, conditioned on a successful fork, $\bar{\chi} \Delta\sn \neq 0$ with non-negligible probability, which yields a contradiction. Hence, this case does not prevent the reduction from succeeding.
    \item $(\delta_k, \delta_s, \delta_e) = (0,0,0)$ and $\sn^\ast = \sn_\pi$. This means the adversary successfully computed the exact scaled responses $\bar{z}_k = \bar{\chi} k_\pi$, $\bar{z}_s = \bar{\chi} s_\pi$, $\bar{z}_e = \bar{\chi} e_\pi$ for a coin it never corrupted. Under Assumption~\ref{ass:lwe}, the commitment $C_\pi$ computationally hides $(k_\pi, s_\pi, e_\pi)$ \cite{lyubashevsky2010ideal,baum2018more}, so the adversary's probability of correctly guessing this exact relation is negligible.
\end{enumerate}
Therefore, if $\varepsilon$ is non-negligible, $\mathcal{B}$ breaks Assumption~\ref{ass:sis} with non-negligible probability $\varepsilon'$, contradicting the hardness of Ring-SIS \cite{ajtai1996generating}. The on-chain ring-existence check with digest matching (Section~\ref{sec:onchain}) enforces condition (i) operationally, confining $\mathbf{C}^\ast$ to genuine deposits and excluding adversarially fabricated commitments with adversary-known openings. \qed
\end{proof}

\subsection{Signer Ambiguity and Unlinkability}

\begin{theorem}[Signer Ambiguity]
\label{thm:anonymity}
In the ROM, under Assumption~\ref{ass:lwe} (commitment hiding), Assumption~\ref{ass:link} (serial-number linking), and the zero-knowledge property of the rejection-sampled sigma protocol, every PPT adversary has negligible advantage in $\mathsf{Exp}^{\mathrm{anon}}$. Consequently no PPT adversary identifies the signer within a ring of size $r$ with $t$ corrupted members with probability exceeding $1/(r-t)$ by more than a negligible amount.
\end{theorem}

\begin{proof}
We prove this via a sequence of indistinguishable games, starting with the real anonymity experiment for a fixed bit $b \in \{0, 1\}$. Let $\mathsf{Adv}_i$ denote the adversary's advantage in $\mathsf{Game}_i$.

\textbf{$\mathsf{Game}_0$:} This is the real $\mathsf{Exp}^{\mathrm{anon}}_b$ experiment. The challenger generates the signature using the exact opening $(k_{i_b}, s_{i_b}, e_{i_b})$ of $C_{i_b}$.

\textbf{$\mathsf{Game}_1$:} We modify the challenger to construct the signature $\sigma$ using the zero-knowledge simulator. Instead of executing the real $\Sig$ algorithm, the challenger randomly samples all responses $z_{k,j}, z_{s,j}, z_{e,j} \leftarrow D_\varsigma$ for all $j \in \{0, \dots, r-1\}$, computes the relation images $t^{\mathrm{pk}}_j$ and $t^{\mathrm{sn}}_j$ directly, and programs the random oracle $\mathcal{H}$ to enforce chain closure. By the honest-verifier zero-knowledge (HVZK) property of the rejection-sampled sigma protocol \cite{lyubashevsky2009fiat,lyubashevsky2012lattice}, the accepted responses in $\mathsf{Game}_0$ are within negligible statistical distance $\Delta$ from the unconditioned mask distribution used in $\mathsf{Game}_1$. Thus, $\lvert \mathsf{Adv}_0 - \mathsf{Adv}_1 \rvert \le \Delta$, which is negligible.
At this point, the signature $\sigma$ is generated without knowledge of the secret index $i_b$ or its witness. The adversary's view depends on $b$ \emph{only} through the published serial number $\sn = a_3 k_{i_b} + e_{i_b}$.

\textbf{$\mathsf{Game}_2$:} We replace the serial number $\sn = a_3 k_{i_b} + e_{i_b}$ with $\sn' = a_3 k' + e'$, where $k', e'$ are freshly sampled independent short elements. Distinguishing $\mathsf{Game}_1$ from $\mathsf{Game}_2$ is exactly the task of deciding whether a serial number shares a secret with a specific commitment. We argue in the random-oracle model, modelling $H_{\mathrm{short}}$ as a random oracle with short outputs. Let $Q_k$ be the event that $\mathcal{A}$ queries $H_{\mathrm{short}}$ at the challenge secret $k_{i_b}$. Conditioned on $\neg Q_k$, the deterministic error $e_{i_b} = H_{\mathrm{short}}(k_{i_b})$ is, in $\mathcal{A}$'s view, a uniform short element independent of $k_{i_b}$ and shared between $C_{i_b}$ and $\sn$; the view is then exactly that of the independent-error variant, in which (by Assumption~\ref{ass:lwe} \cite{lyubashevsky2010ideal,baum2018more}) $C_{i_b}$ hides its opening and $\sn = a_3 k_{i_b} + e_{i_b}$ is a genuine Ring-LWE sample, so $(C_{i_b}, a_3 k_{i_b} + e_{i_b})$ is indistinguishable from $(C_{i_b}, a_3 k' + e')$ by two decisional Ring-LWE hops (the linking statement~\ref{ass:link}). Hence $\lvert \mathsf{Adv}_1 - \mathsf{Adv}_2 \rvert \le 2\varepsilon_{\mathrm{lwe}} + \Pr[Q_k]$, where $\Pr[Q_k] \le Q_H\,\delta$ and $\delta$ bounds the per-query probability of recovering the committed $k_{i_b}$ (search Ring-LWE, equivalently Assumption~\ref{ass:sis} on the coin); all terms are negligible.

In $\mathsf{Game}_2$, the entire view provided to the adversary, including the signature $\sigma$ and the serial number $\sn'$, is completely independent of the bit $b$. Therefore, the adversary's advantage in $\mathsf{Game}_2$ is exactly $0$. Summing the bounds, the adversary's total advantage in $\mathsf{Exp}^{\mathrm{anon}}$ is bounded by $\Delta + \varepsilon_{\mathrm{link}}$, which is negligible.

Consequently, no PPT adversary can identify the signer among the $r-t$ honest, uncorrupted members with probability exceeding $1/(r-t)$ by more than a negligible amount, satisfying the standard anonymity game for linkable ring signatures \cite{liu2004linkable,alberto2018post}. \qed
\end{proof}

\begin{corollary}[Deposit--Withdrawal Unlinkability]
\label{cor:unlink}
Under the same assumptions, no PPT observer links a specific deposit $C$ to a withdrawal beyond the $1/(r-t)$ ring ambiguity: the withdrawal publishes only $(\mathbf{C}, \mu, \sn, \sigma)$, the transcript is simulatable from these, and deciding whether $C$ and $\sn = a_3 k + e$ share the committed $k$ reduces to decisional Ring-LWE in the random-oracle model (Assumption~\ref{ass:lwe}; see Remark~\ref{rem:gaps} for the self-correlated-noise caveat of the deterministic $e = H_{\mathrm{short}}(k)$); recovery of $k$ from $\sn$ is the search Ring-LWE problem.
\end{corollary}

\subsection{Linkability, Non-Frameability, and Protocol-Level Security}

\begin{theorem}[Linkability]
\label{thm:linkability}
Except with negligible probability, every verifying tuple carries $\sn = a_3 k + e$ for the unique short $(k,e)$ opening some ring member, and any two verifying spends of the same coin produce identical nullifiers $\snstar$. In particular, an adversary holding one coin cannot output two verifying withdrawals with distinct nullifiers.
\end{theorem}

\begin{proof}
By the extraction established in Theorem~\ref{thm:unforgeability} (following the linkability framework of \cite{liu2004linkable,alberto2018post}), any verifying tuple implies the existence of an extractor that recovers a short (relaxed) witness $(\bar{z}_k, \bar{z}_s, \bar{z}_e)$ satisfying $a_1 \bar{z}_k + a_2 \bar{z}_s + a_4 \bar{z}_e = \bar{\chi} C_\pi$ alongside the serial-number relation $a_3 \bar{z}_k + \bar{z}_e = \bar{\chi}\,\sn$ for the published $\sn$.

By the binding property of the commitment (Assumption~\ref{ass:sis}), the short opening $(k, s, e)$ of $C_\pi$ is unique. As shown in the proof of Theorem~\ref{thm:unforgeability}, this implies $\bar{z}_k = \bar{\chi} k$ and $\bar{z}_e = \bar{\chi} e$ except with negligible probability. Substituting these into the serial-number relation yields $\bar{\chi}(a_3 k + e - \sn) = 0$. Because $\sn$ is committed to before the random oracle challenge is generated, the probability that the challenge difference $\bar{\chi}$ annihilates a non-zero $\Delta\sn = a_3 k + e - \sn$ is strictly bounded away from 1. Thus, conditioned on a successful extraction, any verifying transcript must have $\sn = a_3 k + e$.

Because $\sn = a_3 k + e$ is a deterministic function of that unique opening, the serial number attached to a given coin is permanently fixed. Since the nullifier $\snstar$ is a deterministic collision-resistant digest of $\NTT(\sn)$, it is also fixed. Therefore, any two verifying spends of the same coin must produce the identical nullifier $\snstar$, making double-spends trivially detectable. \qed
\end{proof}

\begin{lemma}[Non-Frameability]
\label{lem:nonframe}
No PPT adversary can produce a verifying withdrawal carrying the serial number of an honest user's unspent coin, except with negligible probability.
\end{lemma}

\begin{proof}
A framing attempt requires the adversary to broadcast a verifying withdrawal carrying the serial number $\sn = a_3 k + e$ belonging to an honest, unspent coin, aligning with the non-frameability (non-slanderability) definitions formalized in \cite{goodell2019concise}. By Theorems~\ref{thm:unforgeability} and \ref{thm:linkability}, the existence of such a verifying tuple implies that the adversary implicitly knows a short (relaxed) opening of the honest coin. Because the coin is uncorrupted and its secret $(k, s, e)$ was never revealed, constructing this signature constitutes a successful forgery under the strict rules of $\mathsf{Exp}^{\mathrm{uf}}$. By Theorem~\ref{thm:unforgeability}, this allows a reduction to break Assumption~\ref{ass:sis} (Ring-SIS) or Assumption~\ref{ass:lwe} (Ring-LWE). Therefore, a PPT adversary has negligible probability of pre-spending an honest coin or maliciously publishing its nullifier to cause a future honest withdrawal to be rejected as a double-spend. \qed
\end{proof}

\begin{lemma}[Replay Resistance and Context Binding]
\label{lem:replay}
A verifying withdrawal cannot be replayed, and none of recipient, relayer, fee, or application identifier can be altered after signing, except with negligible probability.
\end{lemma}

\begin{proof}
The settlement context $\mu$ (comprising the recipient, relayer, fee, and application ID) prefixes the input to every Fiat--Shamir chain-update hash $\mathcal{H}$ \cite{fiat1986prove,bellare1993random}. Altering even a single bit of any component in $\mu$ after the signature has been generated cascades through the random oracle, completely randomizing the output challenges. By Lemma~\ref{lem:challenge}, the mutated challenge chain will fail to close back to $c_0$ except with negligible probability $\approx 2^{-256}$. Consequently, an adversary attempting to redirect the funds or alter the fee must forge an entirely new signature, which is computationally infeasible under Theorem~\ref{thm:unforgeability}.

Operationally, the smart contract tightly couples verification and payout by settling only to the specific $\mu$ values recorded in the chain-bound state box. This closes the gap between signature verification and payout construction, ensuring the AVM strictly enforces the cryptographically authorized intent. Finally, replay of a successfully executed withdrawal is prevented unconditionally by the flat $O(1)$ state check against the nullifier box, and the inclusion of the \texttt{appID} securely isolates the proof from cross-deployment replay attacks. \qed
\end{proof}

\begin{corollary}[Double-Spend Resistance of the Protocol]
\label{cor:doublespend}
Each deposit is withdrawable at most once: the contract asserts the absence of $\texttt{"n"} \Vert \snstar$ and creates it atomically with the payout, and by Theorem~\ref{thm:linkability} every spend of a given coin presents the identical $\snstar$. AVM group atomicity guarantees that the nullifier is recorded iff the chain closed and funds move iff the nullifier is recorded. Distinct coins map to distinct nullifiers except with negligible probability: if two coins had the same serial number, then $a_3 k_1 + e_1 = a_3 k_2 + e_2$, yielding $a_3(k_1 - k_2) + (e_1 - e_2) = 0$, which is a Ring-SIS collision for $(a_3, 1)$. Furthermore, $\NTT$ is a bijection, so distinct serial numbers give distinct packed images, and a nullifier collision would require a SHA-256 collision, negligible in the random-oracle model; hence one coin's withdrawal cannot pre-register or block another's.
\end{corollary}

\begin{remark}[Abstraction Gaps and Scope of the Proofs]
\label{rem:gaps}
We state explicitly where the arguments are conditional, complementing Section~\ref{sec:limitations}.
\begin{itemize}
    \item \textbf{Relaxed extraction and soundness gap.} The extractor of Theorem~\ref{thm:unforgeability} recovers a \emph{relaxed} opening with relaxation factor $\bar{\chi}$, not an exact opening, and we do not rely on $\bar{\chi}$ being invertible; in the fully splitting ring ($X^{n}+1$ factors into $n$ linear terms modulo $q = 12289$, since $2n \mid q-1$) short challenge differences are not guaranteed invertible, so techniques that assume invertibility do not apply. This is a direct consequence of the ring choice: MatRiCT \cite{esgin2019matrict} and DualRing-LB \cite{yuen2021dualring} obtain exact (relaxed) openings and full balance/soundness against rings containing adversarial members precisely because they work over \emph{partially} splitting rings in which short challenge differences are invertible \cite{lyubashevsky2018short}, whereas we deliberately adopt the fully splitting Falcon-512 ring so that the on-chain verifier can evaluate every relation through a native negacyclic NTT (Section~\ref{sec:onchain}). The reduction we give avoids inversion by subtracting the honest opening, which is sound for theft resistance and non-frameability but does not upgrade to an exact-opening or knowledge-soundness statement. For the same reason we state $\mathsf{Exp}^{\mathrm{uf}}$ (Definition~\ref{def:uf}) over an all-honest ring: extending theft resistance to a balance-style notion over rings that also contain adversarially generated decoys (as in the MatRiCT model \cite{esgin2019matrict}) is standard in principle but requires the invertible-challenge-difference machinery, which the fully splitting ring does not provide; we leave a dedicated relaxed-soundness analysis for this ring as future work.
    \item \textbf{Deterministic serial-number error (self-correlated noise).} The blinding is $e = H_{\mathrm{short}}(k)$, a \emph{deterministic} function of $k$, so $\sn = a_3 k + e$ is not a standard Ring-LWE sample with an independent error but a \emph{self-correlated-noise} instance. The anonymity and unlinkability results (Theorem~\ref{thm:anonymity}, Corollary~\ref{cor:unlink}) are therefore stated in the random-oracle model: with $H_{\mathrm{short}}$ a random oracle, $e$ acts as an independent short error unless the adversary queries the oracle at the secret $k$, and that event reduces to key recovery (search Ring-LWE, equivalently Assumption~\ref{ass:sis} on the coin). We do \emph{not} claim a black-box reduction to standard decisional Ring-LWE; recovering $k$ from $\sn$ is search Ring-LWE (bounded-distance decoding of $a_3 k + e = \sn$ with short $(k,e)$). A clean standard-model reduction is recovered by the alternative design that samples $e$ \emph{independently} at coin generation and stores it in the coin secret (at the cost of a larger persisted opening); we adopt the hash-derived $e$ because it allows deterministic recovery from a minimal persisted seed and inherits determinism from commitment binding rather than from the hash.
    \item \textbf{Rejection sampling with a hash-derived secret.} Although $k$ and $e = H_{\mathrm{short}}(k)$ are not independent, the Fiat--Shamir-with-aborts zero-knowledge argument is unaffected: the honest-verifier simulator samples every response $z_{k,i}, z_{s,i}, z_{e,i}$ from $D_\varsigma$ and never evaluates $H_{\mathrm{short}}$, so the accepted-response distribution is independent of the secret \emph{regardless} of correlations among $(k,s,e)$; the correlation enters only the acceptance probability (termination), where $e$, short by the random-oracle modelling of $H_{\mathrm{short}}$, is covered by the same norm and $M$-bound analysis as $k$ and $s$. Consequently the joint ``double-constrained'' test ($z_{k,\pi} - \chi_\pi k^\ast$ and $z_{e,\pi} - \chi_\pi H_{\mathrm{short}}(k^\ast)$ both short) conveys no advantage against the idealized full-ratio sampler; against the norm-bound-only sampler actually implemented it inherits only the residual leakage already noted under the sampler simplification, negligible for the single spend a coin admits at $\varsigma = 2500$.
    \item \textbf{Sampler simplification.} The implemented sampler applies a norm-bound-only acceptance test on rounded-Gaussian masks, and does not condition the decoy responses on the bound, rather than the full Fiat--Shamir-with-aborts acceptance ratio ($M \approx 3$) assumed by the zero-knowledge step of Theorem~\ref{thm:anonymity}. The zero-knowledge (hence anonymity) guarantee is therefore exact only for the idealized full-ratio sampler; the correctness statement (Theorem~\ref{thm:correctness}) is correspondingly ``with overwhelming probability'' rather than perfect.
    \item \textbf{Concrete parameters.} The parameter set $(\eta, \varsigma, \beta, w)$ has not been calibrated against lattice-reduction estimates \cite{albrecht2015concrete}, so no numeric bit-security level is claimed; only the ring $(n,q)$ is inherited from Falcon-512.
    \item \textbf{Model.} All reductions are classical-ROM; the deployed post-quantum setting warrants a QROM analysis \cite{boneh2011random,don2019security} of the composed scheme, which we leave open.
\end{itemize}
The structural results are the most robust: uniqueness of the serial number (Theorem~\ref{thm:linkability}), double-spend resistance (Corollary~\ref{cor:doublespend}), and context binding (Lemma~\ref{lem:replay}) hold as stated, relying only on Assumption~\ref{ass:sis} and the random oracle. The quantitative unforgeability and anonymity guarantees are conditional on the items above.
\end{remark}

\bibliographystyle{splncs04}
\bibliography{ref}

\end{document}